\documentclass{amsart}

\usepackage{amssymb,latexsym,amsfonts,amsmath}
\usepackage{graphicx}
\include{diagrams}
\usepackage{eso-pic}
\usepackage{color}
\usepackage{type1cm}

\newtheorem{theorem}{Theorem}[section]
\newtheorem{lemma}[theorem]{Lemma}

\newtheorem{proposition}[theorem]{Proposition}

\newtheorem{remark}[theorem]{Remark}
\numberwithin{equation}{section}

\newcommand{\R}{{\mathbb{R}}}

\newcommand{\N}{{\mathbb{N}}}

\newcommand\numberthis{\addtocounter{equation}{1}\tag{\theequation}}
\newcommand{\norm}[1]{\left\lVert #1 \right\rVert}
\begin{document}
\title{Data-driven Stabilization of\\SISO Feedback Linearizable Systems}

\author[Lucas~Fraile]{Lucas Fraile}
\address{Department of Electrical and Computer Engineering\\ 
University of California at Los Angeles\\
Los Angeles, CA 90095-1594, USA}
\email{lfrailev@ucla.edu}

\author[Matteo~Marchi]{Matteo Marchi}
\address{Department of Electrical and Computer Engineering\\ 
University of California at Los Angeles\\
Los Angeles, CA 90095-1594, USA}
\email{matmarchi@ucla.edu}

\author[Paulo~Tabuada]{Paulo Tabuada}
\address{Department of Electrical and Computer Engineering\\ 
University of California at Los Angeles\\
Los Angeles, CA 90095-1594, USA}
\urladdr{http://www.ee.ucla.edu/$\sim$tabuada}
\email{tabuada@ee.ucla.edu}
\thanks{The first author would like to acknowledge Aaron Ames and Jessy Grizzle for their valuable input when developing the first version of these results reported in~\cite{TMGA17}.}%
\thanks{This work was supported in part by the CONIX Research Center, one of six centers in JUMP, a Semiconductor Research Corporation (SRC) program sponsored by DARPA.}% <-this % stops a space

%% The paper headers
%%\markboth{IEEE TRANSACTIONS ON AUTOMATIC CONTROL,~Vol.~---, No.~---,---~2020}{}
% The only time the second header will appear is for the odd numbered pages
% after the title page when using the twoside option.
% 
% *** Note that you probably will NOT want to include the author's ***
% *** name in the headers of peer review papers.                   ***
% You can use \ifCLASSOPTIONpeerreview for conditional compilation here if
% you desire.

% If you want to put a publisher's ID mark on the page you can do it like
% this:
%\IEEEpubid{0000--0000/00\$00.00~\copyright~2015 IEEE}
% Remember, if you use this you must call \IEEEpubidadjcol in the second
% column for its text to clear the IEEEpubid mark.

% use for special paper notices
%\IEEEspecialpapernotice{(Invited Paper)}

\begin{abstract}
In this paper we propose a methodology for stabilizing single-input single-output feedback linearizable systems when no system model is known and no prior data is available to identify a model. Conceptually, we have been greatly inspired by the work of Fliess and Join on intelligent PID controllers (e.g.,~\cite{FJ09,MFC13}) and the results in this paper provide sufficient conditions under which a modified version of their approach is guaranteed to result in asymptotically stable behavior. One of the key advantages of the proposed results is that, contrary to other approaches to controlling systems without a model (or with a partial model), such as reinforcement learning, there is no need for extensive training nor large amounts of data. Technically, our results draw heavily from the work of Nesic and co-workers on observer and controller design based on approximate models~\cite{AN04,NT04}. Along the way we also make connections with other well established results such as high-gain observers and adaptive control. Although we focus on the simple setting of single-input single-output feedback linearizable systems we believe the presented results are already theoretically insightful and practically useful, the last point being substantiated by experimental evidence.
\end{abstract}

% make the title area
\maketitle

% As a general rule, do not put math, special symbols or citations
% in the abstract or keywords.

% For peer review papers, you can put extra information on the cover
% page as needed:
% \ifCLASSOPTIONpeerreview
% \begin{center} \bfseries EDICS Category: 3-BBND \end{center}
% \fi
%
% For peerreview papers, this IEEEtran command inserts a page break and
% creates the second title. It will be ignored for other modes.

\section{Introduction}
\subsection{Motivation}
This paper was motivated by two initially independent lines of inquiry: the thought-provoking work of Fliess and Join on intelligent PID controllers~\cite{FJ09,MFC13}, and the growing impact of machine learning, in particular deep learning, on a wide variety of engineering problems~\cite{LeCunDeepLearning,SchmidhuberDeepLearning}. Curiously, the techniques of Fliess and Join can be seen as a method to transform sensor measurement data into control inputs with minimal reliance on plant models. Therefore, we can interpret intelligent PID controllers as data-driven\footnote{The term \emph{model-free} is sometimes used in lieu of \emph{data-driven}. However, we know from behavioral systems theory that data generated by interacting with a system, i.e., its behavior, is essentially a model for such system. Therefore, we find the term data-driven more adequate as it only suggests that state-space models are not explicitly used.} controllers and this is the view espoused in this work. 

\subsection{Contribution}

The main contribution of this paper is the identification of a class\footnote{Essentially single-input single-output feedback linearizable systems, see Section~\ref{Sec:Main} for a formal statement of the main results. Note, however, that the results conceptually extend to multiple-input, multiple-output systems and even to slowly time varying systems.} of nonlinear systems for which a modified version of intelligent PID controllers can guarantee asymptotic stability. This is by no means the largest class of such systems, but a large enough class to make the technical contribution of this paper relevant to applications, as illustrated by the experimental results presented in Section~\ref{Sec:Experimental}. Moreover, the techniques used to prove the results are also of interest as they rely on an apparently unrelated line of work by Nesic and co-workers~\cite{AN04,NT04} on state estimation and control based on approximate models. In particular, we show in this paper how the results in~\cite{AN04,NT04} can be used to provide a formal justification for the working assumption upon which the analysis of Fliess and Join~\cite{FJ09,MFC13} relies: \emph{the sampling rate can be made high enough so that the relevant signals can be considered constant in between sampling instants}.

Although the use of learning techniques has been surging\footnote{As revealed, e.g., by a search using the keywords ``data-driven'' and ``control''.} within the control community, learning has always been an integral part of the scientific discipline of control. Classical bodies of work within control, such as a system identification~\cite{SI,SISub} and adaptive control~\cite{StableAdaptiveSystemsBook,RobustAdaptiveControlBook}, are essentially learning techniques tailored to the needs of control. The results in this paper make connections with, and sometimes have been inspired by, such classical results. We shall expose several of these connections throughout the paper although readers with a different background may see other connections that have eluded the authors. Yet, it matters to highlight the advantages of the results in this paper over other learning techniques for control. First, the proposed data-driven controllers require neither large amounts of data nor lengthy offline or online training. In this sense, they are much closer to adaptive control than to techniques based on reinforcement learning~\cite{RL18} or deep learning~\cite{DL16}. However, contrary to most work on adaptive control that relies on linearly parameterized models (for the plant or controller), the proposed data-driven controllers do not attempt to learn parameters and, instead, directly learn the input to be fed to the plant. Hence, we always work on small finite-dimensional spaces and, for this reason, only need small amounts of data. A further advantage of the proposed data-driven controllers is that its users only need, yet are not restricted, to employ linear control techniques, an observation that justifies the well crafted title of~\cite{FJ09}. Finally, the results in this paper should be regarded as a design methodology since its key steps can be performed by resorting to different techniques. To show feasibility of the approach, and ease of use, we propose a specific technique for each step although it should be clear these are by no means unique or even the best. We shall return to this point in more detail in Section~\ref{Sec:Methodology} where we provide an outline of the proposed data-driven control methodology.
 It is worth mentioning that the presented methodology results in asymptotically stable behavior without resorting to persistency of excitation assumptions. This is a key contribution, setting us apart from most adaptive control techniques, since it is often hard to justify or validate persistency of excitation in practical applications. 

We would be remiss if we did not give due importance to the limitations of the proposed data-driven control methodology: it can be quite sensitive to measurement noise. This is a consequence of the need to estimate derivatives of sensed signals. While we leave a detailed study of how to best handle noise for future work, the experimental results in Section~\ref{Sec:Experimental} already offer evidence that the proposed data-driven methodology can be practically useful despite the aforementioned limitation. 

\subsection{Related work} As previously stated, the results in this paper were directly inspired by the work of Fliess and Join on intelligent PID controllers. We regard the papers~\cite{FJ09,MFC13} as entry points into this literature since the number of papers on this topic has been growing over the last ten years. The main contributions with respect to this line of work are: 1) to rigorously formalize the idea that signals can be treated as constant in between sampling times provided the sampling rate is high enough; 2) to identify a class of nonlinear systems for which this type of data-driven controllers is guaranteed to result in asymptotically stable behavior. This was accomplished by: 1) proposing several modifications to intelligent PID controllers; 2) a feedback linearizability assumption; and 3) leveraging the work of Nesic and co-workers on estimation and control based on approximate models. Moreover, we also address the case where the control gain is unknown whereas it is assumed to be known in the intelligent PID literature. Although we focus on the simple case of single-input single-output systems, the attentive reader will notice the results can be generalized to multiple-input multiple-output, and partially feedback linearizable systems. We discuss such extensions in Section~\ref{Sec:Main}.

Two recent papers~\cite{DeePC,PT20}, inspired by behavioral techniques, have also proposed data-driven control techniques. It is shown, in both cases, that the proposed controllers can be used with nonlinear systems even though they were developed for linear systems. The key requirement is that the mismatch between the linear and nonlinear models is small. A similar idea is used in this paper: by choosing a suitably high sampling rate, a point-wise linear approximation suffices for control. For this reason the authors suspect it may be possible to combine these different perspectives to obtain even stronger results. The use of behavioral techniques for the development of data-driven control techniques is not recent and had been advocated before, see~\cite{MR08,MR09}. However the algorithms proposed in this earlier work are better suited for offline computation as they require several complex matrix operations. All the aforementioned papers, as well as~\cite{BKP19}, rely on acquiring enough sufficiently informative data to produce control inputs (see~\cite{waarde2019data} for a discussion on how much informative data is required for different control tasks). This requires that enough experiments are conducted using persistently exciting inputs. In contrast, no prior data or persistency of excitation is required for the results in this paper. 

The previous observation sets the current paper apart from much work on data-driven control as well as other work that, although was not developed under the recent data-driven perspective, can be interpreted as such. One such example is the use of extremum seeking ideas, originally developed for optimization purposes, for stabilization, see \cite{KrsticES}. Extremum seeking relies on persistent high-frequency perturbations to estimate gradients and for this reason it is only possible to establish practical stability with this technique. In this line of work, persistency of excitation is typically not stated as an assumption since it is enforced by incorporating high-frequency signals into the input. %One of the advantages of extremum seeking is that it applies to a wider class of systems as we describe in detail in Section~\ref{Sec:Models}.}

Another example, is the control of nonlinear systems using Euler approximations that are learned in real-time, see \cite{Mareels}. This line of work bears some similarities with the approach described in this paper. A key difference is that, while in \cite{Mareels} an approximate plant model is learned, in this paper we directly learn the input to be applied to the plant, which provides the benefit of not requiring knowledge of upper and lower bounds on the control gain. Furthermore, the results in \cite{Mareels} rely again on persistency of excitation which is enforced by design and, for this reason, cannot guarantee asymptotic stability but rather practical stability. % Similarly to extremum seeking, this line of work also applies to a wider class of systems as we discuss in Section~\ref{Sec:Models}.}

The attentive reader might also find some similarities between the approach presented in this paper and Khalil's work on extended high-gain observers and feedback control via disturbance compensation \cite{HighGainObservers}. On the one hand, both of these approaches seek to guarantee the observer's and controller's dynamics are sufficiently fast relative to the plant's dynamics. On the other hand, this objective is achieved in very different ways. While in \cite{HighGainObservers} the key technical idea is the use of high gains to ``speed up'' the controller's dynamics with respect to the plant's, the proposed data-driven controllers ``slow down'' the plant's dynamics through high frequency sampling. By dispensing with the need for high gains, our data-driven approach becomes exempt from the peaking phenomenon, thereby not requiring saturation of the input or state estimates.
%There is, yet, another distinguishing feature with respect to other data-driven control approaches that directly work with unknown nonlinear functions, see, e.g.,~\cite{RADP,RLOFC} and the literature cited therein. The data-driven controllers proposed in this paper only learn the value of a few nonlinear functions at the current state whereas \r{almost all} the previously cited literature seeks to directly learn nonlinear functions or approximations thereof.  A notable exception is \cite{Mareels} that also leverages the idea of learning a few parameters describing the Euler approximation of a the nonlinear system being controlled.

Preliminary versions of the results in this paper appeared in the conference publications~\cite{TMGA17,TF19}. While in~\cite{TMGA17} the control gain is assumed to be known this assumption was dropped in~\cite{TF19}. However, the results in~\cite{TF19} rely on a persistency of excitation assumption that, as previously mentioned, is difficult to verify in practice. In this paper we assume neither the control gain to be known (although we assume knowledge of its sign) nor persistency of excitation.
\color{black}

\section{Notation}
\subsection{Miscellanea}
The natural numbers, including zero, are denoted by $\N$, the real numbers by $\R$, the non-negative real numbers by $\R_0^+$, and the positive real numbers by $\R^+$. If $c:\R\to \R^n$ is a function of time, we denote its first time derivative by $\dot{c}$. When higher time derivatives are required, we use the notation $c^{(k)}$ defined by the recursion $c^{(1)}=\dot{c}$ and $c^{(k + 1)}=\left(c^{(k)}\right)^{(1)}$. The Lie derivative of a function $h:\R^n\to \R$ along a vector field $f:\R^n\to\R^n$, given by $\frac{\partial h}{\partial x}f$, is denoted by $L_fh$. 

Given a symmetric matrix $Q$ we denote by $\lambda_{\min}(Q)$ its smallest eigenvalue and by $\lambda_{\max}(Q)$ its largest eigenvalue.
\subsection{Big O notation}

Consider a function $f:\R_0^+\times \mathcal{Q}\to \R^n$ with $\mathcal{Q}\subseteq \R^n$. We will use the notation $f(t,x)=O_x(T)$ to denote the existence of constants $M,T\in \R^+$ so that for all $t\in [0,T]$ and $x\in \mathcal{Q}$ we have $\Vert f(t,x)\Vert \le MT\Vert x\Vert$ with $\Vert x\Vert$ denoting the 2-norm of $x$. Going forward we will only consider $T\le 1$, thus the following rules apply to this notation where the equalities below are to be used to replace the left-hand side with the right-hand side:
\begin{align*}
O_x(T^2)=O_x(T),\quad \left(O_x(T)\right)^2=O_{x^2}(T^2), \quad TO_x(T)=O_x(T^2),\quad g(x)O_x(T)=O_x(T).
\end{align*}

The subscript $x^2$ in $O_{x^2}(T^2)$ indicates we are squaring the norm, i.e., $O_{x^2}(T^2)$ denotes the upper bound $MT^2\Vert x\Vert^2$. Moreover, the function  $g$ is assumed to have bounded norm, i.e., there exists $b\in \R^+$ so that $\Vert g(x)\Vert\le b$ for all $x\in \mathcal{Q}$. To illustrate the use of these equalities, consider the equality $f(t,x)=O_x(T^2)$ which is defined by $\Vert f(t,x)\Vert\le MT^2\Vert x\Vert$. Given that we chose $T\le 1$, we have the bound $T^2\le T$ that enables us to conclude $\Vert f(t,x)\Vert\le MT\Vert x\Vert$, i.e., $f(t,x)=O_x(T)$. Using the above rules we can directly replace $f(t,x)=O_x(T^2)$ with $f(t,x)=O_x(T)$.

\section{Models}
\label{Sec:Models}
We consider an unknown single-input single-output nonlinear system described by:
\setlength{\arraycolsep}{0.0em}
\begin{eqnarray}
\label{UnknownSystem1a}
\dot{x}&{}={}&f(x) + g(x)u\\
\label{UnknownSystem1b}
y&{}={}&h(x) + d,
\end{eqnarray}
\setlength{\arraycolsep}{5pt}
where $f:\R^n\to\R^n$, $g:\R^n\to\R^n$, and $h:\R^n\to \R$ are smooth functions and  we denote by $y\in \R$, $x\in \R^n$, $u\in \R$, $d\in\R$, the output, state, input, and measurement noise, respectively. We make the assumption that the output function $h$ has relative degree $n$, i.e., this system is feedback linearizable. This means that $L_gL_f^ih(x)=0$ for \mbox{$i=0,\hdots,n - 2$} and $L_gL_f^{n - 1}h(x)\ne 0$ for all $x\in \R^n$. Since the function $L_gL_f^{n - 1}h$ is continuous and never zero, its sign is constant. We will assume the sign of $L_gL_f^{n - 1}h$ to be known and, without loss of generality, take it to be positive. Knowledge of the sign of $L_gL_f^{n - 1}h$ is not a strong assumption beyond $L_gL_f^{n - 1}h\ne 0$. A simple input/output experiment can be performed to infer the sign of $L_gL_f^{n - 1}h$. While requiring some knowledge of the control gain in the form of bounds is standard practice when handling unknown systems, such is the case in, e.g.,~\cite{Mareels},~\cite{KrsticES}, our approach is free of such an assumption. 
%The attentive reader will notice the results in this paper can be generalized to linear time-varying, multiple-input multiple-output, and input-output feedback linearizable systems and we discuss such extensions in Section~\ref{Sec:Main}.}

%\b{To be specific, Krstic requires bounds on the mean-squared value of $\alpha$ but I dont think its worth pointing that detail in here.}

With the objective of presenting the results in its most understandable form, we assume $n=2$ throughout this paper, although all the results hold for arbitrary $n\in \N$. This will enable us to perform all the necessary computations explicitly and without the need for distracting bookkeeping. To further reduce bookkeeping, we will perform most of the analysis under the assumption of noise free measurements (i.e., $d=0$), which will lead to our main result, Theorem \ref{Theorem1}. Given that this assumption does not usually hold when working with physical systems, we also provide Theorem \ref{Theorem2} which establishes stability guarantees under essentially bounded measurement noise. 

Invoking the feedback linearizability assumption, we can rewrite the unknown dynamics in the coordinates \mbox{$(z_1,z_2)=\Psi(x)=(h(x),L_fh(x))$}:
\setlength{\arraycolsep}{0.0em}
\begin{eqnarray}
\label{UnknownSystem2a}
\dot{z}_1&{}={}&z_2\\
\label{UnknownSystem2b}
\dot{z}_2&{}={}&\alpha(z) + \beta(z)u\\
\label{UnknownSystem2c}
y &{}={}& z_1,
\end{eqnarray}
\setlength{\arraycolsep}{5pt}
where $\alpha=L_f^2h\circ\Psi^{-1}$ and $\beta=L_gL_fh\circ\Psi^{-1}$. We note that $f$, $g$, and $h$ are unknown and thus so are $\alpha$ and $\beta$. This form of the dynamics has the advantage of using the two scalar valued functions $\alpha$ and $\beta$ to describe the full dynamics, independently of the value of $n$. This is a key observation that underlies the claim that the results below hold for arbitrary $n\in \N$.

System~\eqref{UnknownSystem2a}-\eqref{UnknownSystem2c} will be controlled using piece-wise constant inputs for a sampling time $T\in \R^+$. This means that inputs $u:\R_0^+ \to \R$ satisfy the following equality for all $k\in \N$:
$$u(kT + \tau)=u(kT ),\qquad \forall \tau\in [ 0,T[.$$ 
It will be convenient to use $u$ to denote an input only defined on $[0,T[$. Since the curve $u$ is constant on the interval $[0,T[$, we identify it with the corresponding element of $\R$.

The solution of~\eqref{UnknownSystem2a}-\eqref{UnknownSystem2b} is denoted by \mbox{$F_t^e(z,u)=(F_{t,1}^e(z,u),F_{t,2}^e(z,u))$}, for $t\in [0,T[$, and satisfies $F_0^e(z,u)=z$. The superscript ``$e$'' reminds us that this is an exact solution. In the next section we discuss approximate solutions.

\section{Approximate models}
In this section we develop an approximate solution of ~\eqref{UnknownSystem2a}-\eqref{UnknownSystem2b} based on the well known Taylor's theorem that we now recall.
\begin{theorem}[See \cite{RudinBook}]
Let $c:\mathcal{I}\to \R^n$ be an $n$ times differentiable function where $\mathcal{I}\subseteq \R$ is an open and connected set. For any $t,\tau\in \mathcal{I}$ such that $\tau + t\in \mathcal{I}$ we have:
\begin{equation}
\label{Taylor}
c(\tau + t){}={}c(\tau) + c^{(1)}(\tau)t + c^{(2)}(\tau)\frac{t^2}{2} + \hdots + c^{(n - 1)}(\tau)\frac{t^{n - 1}}{(n - 1)!} + c^{(n)}(\tau')\frac{t^{n}}{n!}, 
\end{equation}
for some $\tau'\in [\tau,\tau + t]$.
\end{theorem}

Applying this result to $F_{\tau + t,1}^e$ we obtain:
\begin{align*}
F_{\tau + t,1}^e(z,u)=&F_{\tau,1}^e(z,u) + \left(F_{\tau,1}^e\right)^{(1)}(z,u)t + \left(F_{\tau,1}^e\right)^{(2)}(z,u)\frac{t^2}{2} + \left(F_{\tau',1}^e\right)^{(3)}(z,u)\frac{t^3}{3!}.
\end{align*}
If we only retain the first three terms we obtain an approximate solution with an  approximation error given by the magnitude of the (neglected) fourth term. The following result provides a bound for the approximation error in a form useful for the results derived in this paper.

\begin{proposition}
\label{Prop:BoundTaylor}
Let $\mathcal{D}\subset\R^3$ be a compact set. Then, there exist $T\in \R^+$ and $M\in \R^+$ such that:
\begin{equation}
\label{BoundInProp}
\left\Vert \left(F_{\tau',1}^e\right)^{(3)}(z,u)\frac{t^3}{3!}\right\Vert\le M{T}^3\Vert (z,u - u_0) \Vert,
\end{equation}
for all $(z,u)\in \mathcal{D}$, all $t,\tau'\in [0,T]$, and where \mbox{$u_0= -\beta^{-1}(0)\alpha(0)$}.
\end{proposition}

Using the $O$ notation, this result states that:
$$\left(F_{\tau',1}^e\right)^{(3)}(z,u)\frac{t^3}{3!}=O_{(z,u - u_0)}(T^3).$$

\begin{proof}
Since~\eqref{UnknownSystem2a}-\eqref{UnknownSystem2b} is a smooth differential equation (recall that inputs are constant), solutions exist for all $\tau\in [0,T_{z,u}[$ where $[0,T_{z,u}[$ is the maximal interval for which the solution $F_{\tau,1}^e(z,u)$ exists. The function $(z,u)\mapsto T_{z,u}$ is lower semi-continuous and, given that $(z,u)$ belongs to the compact set $\mathcal{D}$, it achieves its minimum on $\mathcal{D}$. Let $T\in \R^+$ be smaller than $\min_{(z,u)\in \mathcal{D}}T_{z,u}$. By definition of $T$, for any $(z,u)\in \mathcal{D}$ solutions exist on the interval $[0,T]$. Consider now the function $\left(F_{\tau',1}^e\right)^{(3)}$ and note it is continuously differentiable, by assumption, and thus Lipschitz continuous on $\mathcal{D}\times \{\tau'\}$ for each fixed $\tau'\in [0,T]$. Hence, by definition of Lipschitz continuity we have:
\setlength{\arraycolsep}{0.0em}
\begin{equation}
\label{Eq:Lips}
\left\Vert \left(F_{\tau',1}^e\right)^{(3)}(z,u)\right. -\left. \left(F_{\tau',1}^e\right)^{(3)}(z',u')\right\Vert \le L(\tau')\Vert (z,u) - (z',u')\Vert
\end{equation}
\setlength{\arraycolsep}{5pt}
for all $(z,u),(z',u')\in \mathcal{D}$ and all $\tau'\in [0,T]$. Noting that, according to~\eqref{UnknownSystem2a}-\eqref{UnknownSystem2b},  $F_{\tau'}^e(0,u_0)=0$ for $u_0= -\beta^{-1}(0)\alpha(0)$ and all $\tau' \in [0,T]$, we conclude that $\left(F_{\tau',1}^e\right)^{(3)}(0,u_0)=0$.
Using this equality in~\eqref{Eq:Lips} we obtain:
\begin{align*}
\left\Vert \left(F_{\tau',1}^e\right)^{(3)}(z,u)\right\Vert \le L(\tau')\Vert (z,u) - (0,u_0)\Vert=L(\tau')\Vert (z,u - u_0)\Vert,
\end{align*}
by setting $z'=0$ and $u'= - u_0$. If we now take \mbox{$M=\frac{1}{3!}\max_{\tau' \in[0,T]}L(\tau')$} we obtain the desired inequality. Note that $M$ is well defined since $L$ is continuous and $[0,T]$ compact.
\end{proof}

Based on Proposition~\ref{Prop:BoundTaylor} we can write the exact solution $F_t^e$ of~\eqref{UnknownSystem2a}-\eqref{UnknownSystem2b} valid for all $t\in [0,T[$, as:
\setlength{\arraycolsep}{0.0em}
\begin{eqnarray}
\label{ModelOrder3a}
F_{t,1}^e(z,u)&{}={}&z_1 + z_2t + (\alpha(z) + \beta(z)u)\frac{t^2}{2} + O_{(z,u - u_0)}(T^3) \\
\label{ModelOrder3b}
F_{t,2}^e(z,u)&{}={}&z_2 + (\alpha(z)  +  \beta(z)u)t + O_{(z,u - u_0)}(T^2).
\end{eqnarray}
\setlength{\arraycolsep}{5pt}
%Note that $\alpha$ and $\beta$ can be treated as constants since it follows from Proposition~\ref{Prop:BoundTaylor} that:
%\begin{equation}
%\label{Constants}
%\alpha(t)=\alpha(0)+O_{(z,u - u_0)}(T),\qquad \beta(t)=\beta(0)+O_{(z,u - u_0)}(T).
%\end{equation}
%For simplicity of notation, we drop the argument of $\alpha$ and $\beta$ and write:
%\begin{eqnarray}
%\label{ModelOrder23a}
%F_{t,1}^e(z,u)&=&Z_0+z_2t+(\alpha+\beta u)\frac{t^2}{2}+O_{(z,u - u_0)}(T^3)\\
%\label{ModelOrder23b}
%F_{t,2}^e(z,u)&=&z_2+(\alpha+\beta u)t+O_{(z,u - u_0)}(T^2).
%\end{eqnarray}
By setting\footnote{Although $t\in [0,T[$, solutions are not altered by changing the input on a zero measure set.} $t$ equal to $T$, the previous model provides a \emph{family} of discrete-time \emph{approximate} models indexed by $T$:
\setlength{\arraycolsep}{0.0em}
\begin{eqnarray}
\label{ModelOrder33a}
z_1(k + 1)&{}={}&z_1(k) + z_2(k)T + (\alpha(k) + \beta(k) u(k))\frac{T^2}{2} \\
\label{ModelOrder33b}
z_2(k + 1)&{}={}&z_2(k) + (\alpha(k) + \beta(k) u(k))T,
\end{eqnarray}
\setlength{\arraycolsep}{5pt}
where $z(k), \alpha(k)$, and $\beta(k)$ denote the value of $z$, $\alpha(z)$, and $\beta(z)$ at time $kT$, $k\in \N$, respectively. For later use we introduce the notation:
\begin{eqnarray}
F_{T,1}^a(z,u)&{}\stackrel{def.}{=}{}&z_1(k) + z_2(k)T + (\alpha(k) + \beta(k) u(k))\frac{T^2}{2}\notag \\
F_{T,2}^a(z,u)&{}\stackrel{def.}{=}{}&z_2(k) + (\alpha(k) + \beta(k) u(k))T, \notag
\end{eqnarray}
where the superscript ``$a$'' emphasizes the fact that $z$ is the solution of an approximate model.
\color{black}
%
%
%When we need to refer to the exact solution of~\eqref{UnknownSystem2a}-\eqref{UnknownSystem2b} we use the notation $F_{kT}^e(z,u)$ or simply $z(k)$.

%while the notation $F_T^a(z,u)$, or simply $z(k+1)$, is reserved for the solution of the approximate discrete-time model~\eqref{ModelOrder33a}-\eqref{ModelOrder33b}.
\section{A data-driven control design methodology}
\label{Sec:Methodology}
In this section we summarize the proposed data-driven control design methodology that is presented in detail in Sections~\ref{Sec:StateEstimation} and Section~\ref{Sec:ControllerDesign}. The design will be based on different approximate models, all of which are based on the discrete-time approximate model~\eqref{ModelOrder33a}-\eqref{ModelOrder33b}. We start by observing that the model~\eqref{ModelOrder33a}-\eqref{ModelOrder33b} is affine and thus all the design techniques described in this paper only require knowledge of linear systems theory. 

The affine nature of the model~\eqref{ModelOrder33a}-\eqref{ModelOrder33b} suggests that we could use the preliminary controller:
\begin{equation}
\label{preliminary}
\overline{u}(k)=\beta^{-1}(z(k))( - \alpha(z(k))  + v(z(k)),
\end{equation}
where $v(z)$ is a new input, to cancel the effect of the nonlinear functions $\alpha$ and $\beta$ provided that $z(k)$ and the values of $\alpha$ and $\beta$ at the current state $z(k)$ were known. After this preliminary controller, it would be easy to design a virtual controller stabilizing the resulting linear system with input $v$. As an example design technique, we show in Section~\ref{Sec:ControllerDesign} how to design linear controllers that perform this task.

By considering\footnote{Formally justifying this design assumption is one of the purposes of the results in Section~\ref{Sec:Main}.} $\alpha$ and $\beta$ to be constant functions in ~\eqref{ModelOrder33a}-\eqref{ModelOrder33b} we obtain an observable linear system by formally treating $\alpha + \beta u$ as a new state $z_3$ and using the measurement equation $y=z_1$. Hence, any technique to reconstruct the state of an observable linear system can be employed provided the reconstruction error is of order $T$, as specified by equation~\eqref{StateEstimationError} in Section~\ref{Sec:StateEstimation}. As an example design technique, in Section~\ref{Sec:StateEstimation} we propose to reconstruct the state by directly solving the equation $Y=\mathcal{O}z$ where $Y$ is a sequence of measurements and $\mathcal{O}$ is the observability matrix of the aforementioned observable linear system.

Once an estimate of $z_3$ is obtained, we formally treat $z_3$ as an observation. It is well known that reconstructing $\alpha$ and $\beta$ from the measurement equation $z_3=\alpha + \beta u$ is not possible unless a persistency of excitation assumption is placed on the input $u$. Rather than assuming persistency of excitation, we note this type of problem has been extensively studied in adaptive control~\cite{StableAdaptiveSystemsBook,RobustAdaptiveControlBook} and it is known that any choice of parameters $\alpha$ and $\beta$ that satisfies the measurement equation $z_3=\alpha + \beta u$ suffices for control purposes. Inspired by this, we will directly utilize the observation $z_3$ in a dynamic controller generating inputs which asymptotically converge to those generated by our preliminary static controller \eqref{preliminary}.

Once the two aforementioned components -- state estimator and static controller -- have been designed to satisfy the relations~\eqref{StateEstimationError} and~\eqref{DefineV}, it will follow from our main result, Theorem~\ref{Theorem1}, that their concurrent execution, combined with the dynamic controller we provide, will result in asymptotically stable behavior.

%\b{We could comment here about the example in [14] in a roundabout manner (im basically paraphrasing Miroslav), maybe something along the lines of:}

%\b{Our framework can help to bridge the gap often seen between techniques relying on the assumption of the knowledge of exact discrete-time models and the approximate models used in their implementations. Two such examples are [14] and [15]. (To be fair we should still make a comment on that not every technique fits our framework, I think [14] is one such example as it should be made into an online learning problem to start with).}

\section{State estimation}
\label{Sec:StateEstimation}
For state estimation purposes it is convenient to formally treat $\alpha(k) + \beta(k) u(k)$, in the family of approximate models~\eqref{ModelOrder33a}-\eqref{ModelOrder33b}, as the state $z_3$ to obtain:
\setlength{\arraycolsep}{0.0em}
\begin{eqnarray}
\label{ApproxModel3States1}
z_1(k + 1) &{}={}& z_1(k) + z_2(k)T + z_3(k)\frac{T^2}{2}\\
\label{ApproxModel3States2}
z_2(k + 1) &{}={}& z_2(k) + z_3(k)T\\
\label{ApproxModel3States3}
z_3(k + 1)&{}={}&z_3(k).
\end{eqnarray}
\setlength{\arraycolsep}{5pt}
Note that this \emph{approximate} model states that $z_3$ is constant although $\left(F_{t,1}^e\right)^{(2)}$ will, in general, not be so. Equality~\eqref{ApproxModel3States3} follows from applying Proposition~\ref{Prop:BoundTaylor} to $\left(F_{t,1}^e\right)^{(2)}$ and dropping the error term $O_{(z,u - u_0)}(T)$. Since~\eqref{ApproxModel3States1}-\eqref{ApproxModel3States3} is a linear model, it can be written in the form:
$$z(k + 1)=Az(k),\qquad y(k)\stackrel{def.}{=} z_1(k)=Cz(k).$$ 
Moreover, it can be easily checked that $A$ is invertible and we thus denote by $\mathcal{O}$ the observability matrix for the pair $(A^{-1},C)$ which allows us to write:
\begin{equation}
\label{DefY}
Y(k)\stackrel{def.}{=}\begin{bmatrix}y(k)\\y(k - 1)\\\vdots\\y(k - \rho + 1)\end{bmatrix}=\mathcal{O}z(k),
\end{equation}
where $\rho\in \N$, $\rho\ge n + 1$, is the number of measurements that will be used for state estimation. The estimate $\widehat{z}(k)$ of the state vector $z(k)$ can then be obtained by solving this equation via least-squares:
\begin{equation}
\label{StateEstimateFinal}
\widehat{z}(k)=(\mathcal{O}^{T}\mathcal{O})^{-1}\mathcal{O}^T Y(k).
\end{equation}
Given that equalities~\eqref{ApproxModel3States1},~\eqref{ApproxModel3States2}, and~\eqref{ApproxModel3States3} only hold up to $O_{(z,u - u_0)}(T^3)$, $O_{(z,u - u_0)}(T^2)$, and $O_{(z,u - u_0)}(T)$, respectively, we can easily establish the equality \mbox{$z=\widehat{z} + O_{(z,u - u_0)} (T)$}.
%The following simple computation:
%\begin{eqnarray}
%&&Y=\mathcal{O}z+O_{(z,u - u_0)} \left(T^3\right)\notag\\
%&\implies&\mathcal{O}^TY=\mathcal{O}^T\mathcal{O}z+\mathcal{O}^TO_{(z,u - u_0)} (T^3)\notag\\
%&=&\mathcal{O}^TY=\mathcal{O}^T\mathcal{O}z+O_{(z,u - u_0)} \left(\begin{bmatrix}T^3\\T^4\\T^5\end{bmatrix}\right)\notag\\
%&\implies&(\mathcal{O}^T\mathcal{O})^{-1}\mathcal{O}^TY=z+(\mathcal{O}^T\mathcal{O})^{-1}O_{(z,u - u_0)} \left(\begin{bmatrix}T^3\\T^4\\T^5\end{bmatrix}\right)\notag\\
%&\implies&(\mathcal{O}^T\mathcal{O})^{-1}\mathcal{O}^TY=z+O_{(z,u - u_0)} \left(\begin{bmatrix}T^3\\T^2\\T\end{bmatrix}\right)\notag,
%\end{eqnarray}
If we introduce the estimation error $e_z$, defined by $e_z=z - \widehat{z}$, it follows that:
\begin{equation}
\label{StateEstimationError}
e_z=O_{(z,u - u_0)}(T). 
\end{equation}
It is straightforward to show that in the presence of essentially bounded noise on the measurements \eqref{DefY} the state estimation error is given by:
\begin{equation}
\label{StateEstimateNoise}
e_z=O_{(z,u - u _0)}(T) + O_{\,\overline{d}\,}(T^{-n}),\quad \overline{d}\stackrel{def.}{=}\text{ess}\sup_{t\in \R_0^+}\Vert d(t)\Vert,
\end{equation}
where $n$ is the relative degree of the system.

As previously stated, the control scheme proposed in Section~\ref{Sec:ControllerDesign} only depends on the preceding equality. Hence, we can replace least-squares estimation with any other estimation technique leading to~\eqref{StateEstimationError}. In particular, the parameter $\rho$ is not relevant to the theoretical analysis although it will play an important role in mitigating the effect of sensor noise: larger values of $\rho$ ``average out'' the effect of noise. 

\begin{remark}
In~\cite{RJ09} it is shown that the algebraic techniques proposed in~\cite{MJF07}, and used in~\cite{FJ09,MFC13} to estimate derivatives of a measured signal, can be interpreted as estimating the state of the state-space linear model governing the signals $y$ satisfying $y^{(3)}=0$. If we denote the constructability Gramian of this linear model by $W_{cn}$ and its state-transition matrix by $\Phi$, the estimate is given by the well known expression (see (3.9), page 250,~\cite{LinearSystems}):
$$W_{cn}^{-1}\int_{t_0}^{t_1}\Phi^T(\tau,t_1)C^Ty(\tau)d\tau.$$
Equality~\eqref{StateEstimateFinal} can be seen as the discrete-time analogue of this finite-time estimation technique.
\end{remark}
\begin{remark}
The matrix $(\mathcal{O}^{T}\mathcal{O})^{-1}\mathcal{O}^T$ contains terms of the form $T^{-1}$ on its second row and terms of the form $T^{-2}$ on its third row. Hence, it can be conceptually understood as a linear high-gain observer with finite-time convergence and where $T$ plays the role of the parameter $\varepsilon$ used in~\cite{HighGainObservers}. Similarly to high-gain observers, the estimate provided by~\eqref{StateEstimateFinal} can be very sensitive to measurement noise. This can be mitigated by using more samples for estimation so as to ``average out'' noise, i.e., by increasing $\rho$. Contrary to high-gain observers, however, we do not need to explicitly worry about the peaking phenomenon when computing the estimate since it is not computed recursively. As mentioned before,~\eqref{StateEstimateFinal} could be replaced with a high-gain observer or even the more recent low-power high-gain observers~\cite{LowPowerHighGain}. Which specific estimation technique works better in practice, and in the context of the results in this paper, is an important problem that we leave for future research.
\end{remark}

\section{Controller design}
\label{Sec:ControllerDesign}
If we assume the parameters $\alpha$ and $\beta$ to be known, we can design a family of controllers (parameterized by $T$) for the family of approximate models~\eqref{ModelOrder33a}-\eqref{ModelOrder33b} with the objective of  asymptotically stabilizing the origin in the following specific sense: there exists a symmetric and positive definite matrix $P_z$ and constants $\lambda_z,T_0\in \R^+$ so that $V_z(z)=z^TP_z z$ satisfies:
\begin{equation}
\label{DefineV}
V_z(F_T^a(z,\overline{u})) - V_z(z){}\le{} - \lambda_z T\Vert z(k)\Vert^2 + O_{(z,u - u _0)^2}(T^2),
\end{equation}
for all $T$ in the interval $[0,T_0]$. Strikingly, we can achieve this inequality with the very simple family of virtual controllers which is independent of $T$:
\begin{eqnarray}
\label{Controller1}
\overline{u}&=&\beta^{-1}(- \alpha + v(z)),\\
\label{Controller2}
v(z)&=&Kz,
\end{eqnarray}
where $K$ is a suitable matrix. We note that the approximate model~\eqref{ModelOrder33a}-\eqref{ModelOrder33b} can be written as:
\begin{equation}
\label{AffineModel}
F_T^a(z(k),\overline{u})=Az(k) + B\alpha(k) + B\beta(k)\overline{u}(k)=Az(k) + Bv(z(k)),
\end{equation}
where the matrices $A$ and $B$ are of the form:
$$A=I + A_1T,\qquad B=B_1T + B_2 T^2.$$
Since $(A_1,B_1)$ is a controllable pair, there exists a controller $v(z)=Kz$ and a symmetric and positive definite matrix $P_z$ so that:
\begin{equation}
\label{DefineVAgain}
(A_1 + B_1K)^TP_z + P_z(A_1 + B_1K)=  - Q,
\end{equation}
for some symmetric and positive definite matrix $Q$. Using this controller we have:
\begin{equation}
F_T^a(z(k),\overline{u})=(A + BK)z=(I  +  (A_1 + B_1K)T + B_2K T^2)z.\notag
\end{equation}
Computing $V_z(F_T^a(z(k),\overline{u})) - V_z(z)$ provides:
\setlength{\arraycolsep}{0.0em}
\begin{eqnarray}
V_z(F_T^a(z(k),\overline{u})) - V_z(z) &{}={}& z^T(A + BK)^TP_z(A + BK)z - z^TP_z z \notag \\
&{}={}& z^T((A_1 + B_1K)T)^TP_z z + z^TP_z ((A_1 + B_1K)T)z\notag \\
&&{+}\: O_{z^2}(T^2) + O_{z^2}(T^3) + O_{z^2}(T^4)\notag\\
&{}={}& - Tz^TQz + O_{z^2}(T^2) \notag \\
&{}\le{}&  - \lambda_{\min}(Q)T\Vert z\Vert^2 + O_{(z,u - u _0)^2}(T^2),\notag
\end{eqnarray}
\setlength{\arraycolsep}{5pt} 
which is the desired inequality~\eqref{DefineV}.

The dynamics in \eqref{AffineModel} are stated for the preliminary control law $\overline{u}=\beta^{-1}(-\alpha +  v(z))$, yet since neither $\alpha$ nor $\beta$ are known, this controller cannot be directly implemented. Instead, we note that this controller enforces $\alpha+\beta u = v(z)$ and design a dynamic controller that asymptotically enforces this equality by guaranteeing convergence to the origin of the error:
\begin{equation}
\label{eu}
e_u(k) = v(z(k))-\left(\alpha(k)+\beta(k) u(k) \right).
\end{equation}
To achieve this we propose a dynamic control law of the following form:
\begin{equation}
\label{DynamicControllerFinal}
u(k+1)=u(k)+\gamma\left(v(\widehat{z}(k))-\widehat{z}_3(k)\right),
\end{equation}
where $\gamma \in \R^+$ is sufficiently small\footnote{If an upper bound $\overline{\beta}$ for $\beta$ is known, $\gamma<\overline{\beta}^{-1}$ suffices.} In order to fully specify this controller, we need to describe its operation during the initial transient of $\rho - 1$ steps during which enough measurements are collected to produce the first state estimate according to~\eqref{StateEstimateFinal}. We simply choose a fixed sequence of inputs $u^*_0,u^*_1,\hdots,u^*_{\rho - 2}$ to be used during this transient. Although different sequences will lead to different transients, the results in Section~\ref{Sec:Main} are independent of this choice.

The main results in the next section explain why such a dynamic controller works despite being designed for an approximate model while assuming knowledge of the exact values of the parameters and states in its design.

% Computing $V_z((A + BK)z) - V_z(z)$ provides:
%\setlength{\arraycolsep}{0.0em}
%\begin{eqnarray}
%\label{FaVz}
%V_z((A + BK)z) - V_z(z) &{}={}& z^T(A + BK)^TP_z(A + BK)z - z^TP_z z \notag \\
%&{}={}& z^T((A_1 + B_1K)T)^TP_z z + z^TP_z ((A_1 + B_1K)T)z\notag \\
%&&{+}\: O_{z^2}(T^2) + O_{z^2}(T^3) + O_{z^2}(T^4)\notag\\
%&{}={}& - Tz^TQz + O_{z^2}(T^2) \notag \\
%&{}\le{}&  - \lambda_{\min}(Q)T\Vert z\Vert^2 + O_{(z,u - u _0)^2}(T^2),
%\end{eqnarray}
%\setlength{\arraycolsep}{5pt} 
%which is the desired inequality~\eqref{DefineV}.
%\subsection{Implemented controller}
%
%
%%	Our choice of $\eqref{DynamicControllerFinal}$ is backed by the following lemma\footnote{\b{See appendix for proof}}:
%%\begin{lemma}
%%\label{lemma1}
%%The dynamic controller \eqref{DynamicControllerFinal}, where $v(z)\in P_{[r_1,r_2]}$ with $r_1\le r_2\in \N_{\ge 1}$ is a polynomial control law, in conjunction with the Lyapunov function $V_{e_u}=e_u^2$, satisfy \eqref{Definev} with $m_1=2r_1$ and $m_2=2r_2$ when $z$ and $u$ reside in some compact set\footnote{\b{In the proof of Theorem \eqref{Theorem1} we show why we are allowed to make this assumption.}} $R$. 
%%\end{lemma}

\section{Main results}
\label{Sec:Main}

\subsection{The noise-free scenario}
It is pedagogically convenient to start with the noise-free scenario, i.e., $d=0$ in~\eqref{UnknownSystem1b}, as it allows us to expose the key ideas in a simpler manner. Notwithstanding the absence of noise, the proofs of the main results in this section are quite long and for this reason can be found in the Appendix. The authors hope its length does not hide the simple idea upon which it rests: we can formally justify the use of approximate models for observer and controller design by using the frameworks developed by Arcak and Nesic in~\cite{AN04} for the former, and by Nesic and Teel in~\cite{NT04} for the latter. This combination of ingredients shows that for any compact set of initial conditions there exists a sufficiently small sampling time ensuring the proposed controller keeps all the signals bounded and drives the state to the origin. 

\begin{theorem}
\label{Theorem1}
Consider an unknown nonlinear system of the form~\eqref{UnknownSystem1a}-\eqref{UnknownSystem1b} where the output function $h$ has relative degree $n$.  In the absence of measurement noise, i.e., $d=0$, for any compact set $\mathcal{S}\subset \R^n$ of initial conditions containing the origin in its interior there exists a time $T^*\in \R^+$ and a constant $b\in \R^+$ (both depending on $\mathcal{S}$) so that for any sampling time $T\in [0,T^*]$, the dynamic controller \eqref{DynamicControllerFinal}, where the virtual input $v$ is provided by \eqref{Controller2}, using the state estimates provided by an estimation technique satisfying~\eqref{StateEstimationError}, renders the closed-loop trajectories bounded, i.e., $\Vert \widehat{z}(k)\Vert\le b$ and $\Vert e_u \Vert\le b$ for all $k\in \N$, and $\Vert x(t)\Vert\le b$ for all $t\in \R_0^+$. Moreover: 
$$\lim_{t\to\infty} x(t)=0.$$
\end{theorem}

%\b{You told me to adapt the theorem for relative degree n, it looks a bit weird given that the whole paper and the proof are based on n=2. Did I misunderstand you? and if not, are you sure this makes sense without changing the rest of the paper? The same applies for the next theorem.}

Although the previous result only claims that trajectories converge to the origin, it can be readily applied to trajectory tracking problems by considering convergence to zero of the error between the real trajectory and the trajectory to be tracked.
%For technical reasons, found in the proof, we used knowledge of the lower bound $\underline{\beta}$ of $\beta$. The upper-bound $\overline{\beta}$ will only play a role when providing formal results in the presence of disturbances.

Extending these results to MIMO control systems is conceptually simple, with the caveat that $\beta$ and $\gamma$ (now matrices) must be chosen so that the eigenvalues of $I-\beta\gamma$ reside in the unit circle, ensuring convergence of the error $e_u$. An extension to partially feedback linearizable systems is also possible by assuming a well behaved zero dynamics.

We now introduce the following lemma which provides a sufficient condition for the results of Theorem~\ref{Theorem1} to hold under a virtual controller $v$ different from the one provided in \eqref{Controller2}:
\begin{lemma}
\label{Lemma1}
Let the virtual input $v:\R^n \to \R$ be such that the following conditions hold:
\begin{eqnarray}
\label{lemma1req1}
V_z\left(F^a(z(k),\beta^{-1}(k)\left(\alpha(k)+v(z(k))\right)\right)-V_z(z(k))&\leq &- \lambda T\Vert z\Vert^2 + O_{(z,u - u _0)^2}(T^2) \qquad\\
\label{lemma1req2}
v(z(k)+O_{(z,u-u_0)}(T))&=&v(z(k))+O_{(z,u-u_0)}(T),
\end{eqnarray}
where $V_z$ is defined in section \ref{Sec:ControllerDesign}, then the results of Theorem \ref{Theorem1} remain unchanged when using such virtual input in place of the one provided by equation \eqref{Controller2}. 
\end{lemma}

\subsection{The noisy scenario}

As previously mentioned, in the presence of essentially bounded measurement noise, the state estimation error under the state estimation technique described in Section \ref{Sec:StateEstimation} is now given by:
$$e_z=O_{(z,u - u _0)}(T) + O_{\,\overline{d}\,}(T^{-n}),\quad \overline{d}\stackrel{def.}{=}\text{ess}\sup_{t\in \R_0^+}\Vert d(t)\Vert,$$ 
where $\overline{d}$ is the noise bound and $n$ is the relative degree of the system. This expression shines light on the trade-off between choosing a small sampling time to render the approximate models adequate and choosing a large sampling time to reduce the amplification effect on noise. As with the noise-free case, the proof of the following result can be found in the Appendix.

%\b{This is true for the example technique we provide, yet I am not sure it applies for every technique satisfying the requirement \eqref{StateEstimationError}. I will look into it, as if it is not true, the bounds we give in the Theorem below would also depend on the specific technique we use, which would prove problematic from the perspective of providing a proof valid for the framework and not just for the example.}
 
% This expression shines light on the trade-off between choosing a small sampling time to render the approximate models adequate and choosing a large sampling time to reduce the amplification effect on noise.

% As with the noise-free case, the proof of the following result can be found in the Appendix.
 
\begin{theorem}
\label{Theorem2}
Consider an unknown nonlinear system of the form~\eqref{UnknownSystem1a}-\eqref{UnknownSystem1b} where the output function $h$ has relative degree $n$ and assume the noise $d$ to be essentially bounded, i.e., there exists a constant $\overline{d}\in \R^+_0$ satisfying $\overline{d}=\mathrm{ess}\sup_{t\in \R^+_0}\Vert d(t)\Vert$.  For any compact set $\mathcal{S}\subset \R^n$ of initial conditions containing the origin in its interior there exists a time $T^*\in \R^+$ (depending on $\mathcal{S}$), and constants $b_1,b_2,b_3\in \R^+$ (depending on $\mathcal{S}$ and $T^*$) so that for any sampling time $T\in [0,T^*]$, if $\overline{d}\le b_1$, the dynamic controller \eqref{DynamicControllerFinal}, where the virtual input $v$ is provided by \eqref{Controller2}, using the state estimates provided by an estimation technique satisfying~\eqref{StateEstimationError}, renders the closed-loop trajectories bounded, i.e., $\Vert \widehat{z}(k)\Vert\le b_2$ and $\Vert e_u\Vert\le b_2$ for all $k\in \N$, and $\Vert x(t)\Vert\le b_2$ for all $t\in \R_0^+$. Moreover: 
$$\limsup_{t\to\infty} \Vert x(t)\Vert\le b_3\, \overline{d}\, T^{-n}$$
\end{theorem}

\section{Experimental evaluation}
\label{Sec:Experimental}
In this section we report on an experimental evaluation of the proposed data-driven controller to regulate the altitude of a quad-copter. The experiments were performed on a Bitcraze Crazyflie 2.1 and an Optitrack Prime 17W motion capture system was used to measure the quad-copter's altitude during the experiments. An experimental demonstration of the robustness of the proposed data-driven controller is available in the video:\\ \texttt{https://www.youtube.com/watch?v=9EVcRvLOGVo}.

%\b{texttt forces the link to go into the margin, should I put it in the next line?}

\subsection{Experimental  setup}
%\label{SubSec:quad-copter}
The Crazyflie 2.1 is a small open source modular quad-copter designed by Bitcraze AB \cite{Crazyflie}
% This is a platform that has seen widespread adoption through out the research community due mostly to: its low weight of $27$ grams, which allows for safe indoor flight; its modularity, as it permits the fast replacement of damaged parts; and all its firmware being open source, which gives researchers great flexibility of implementation. 
equipped with an IMU based on a 3-axis accelerometer and gyroscope. The baseline firmware for the Crazyflie includes a PID based flight controller. We partitioned this controller into attitude and altitude controllers, keeping the former and replacing the latter with a data-driven controller. 

%\subsection{Optitrack's Prime 17W}
%\label{SubSec:Camera}
To provide the data-driven controller with altitude measurements we used eight Optitrack Prime 17W cameras \cite{Optitrack} distributed on three sides along the top of a roughly cubic area. The cameras have a refresh rate of up to $360$Hz and provide position and pose measurements by triangulating a set of markers placed on the quad-copter. Whereas the PID controller regulating attitude receives measurements from the IMU and the motion capture system, the data-driven controller only receives altitude measurements from the motion capture system.

A qualitative view of the measurement noise, when the quad-copter is static on the floor, is presented in Figure \ref{MeasurementNoise}.  The real altitude corresponds to the location of the markers on top of the quad-copter. We observe the noise typically has a magnitude of 1 mm, i.e., $\overline{d}= 0.001$, although there are occasional troughs in the noise signal corresponding to instants where the motion capture system loses track of some of the markers.
\begin{figure}
\centering
\includegraphics[width=0.486\textwidth]{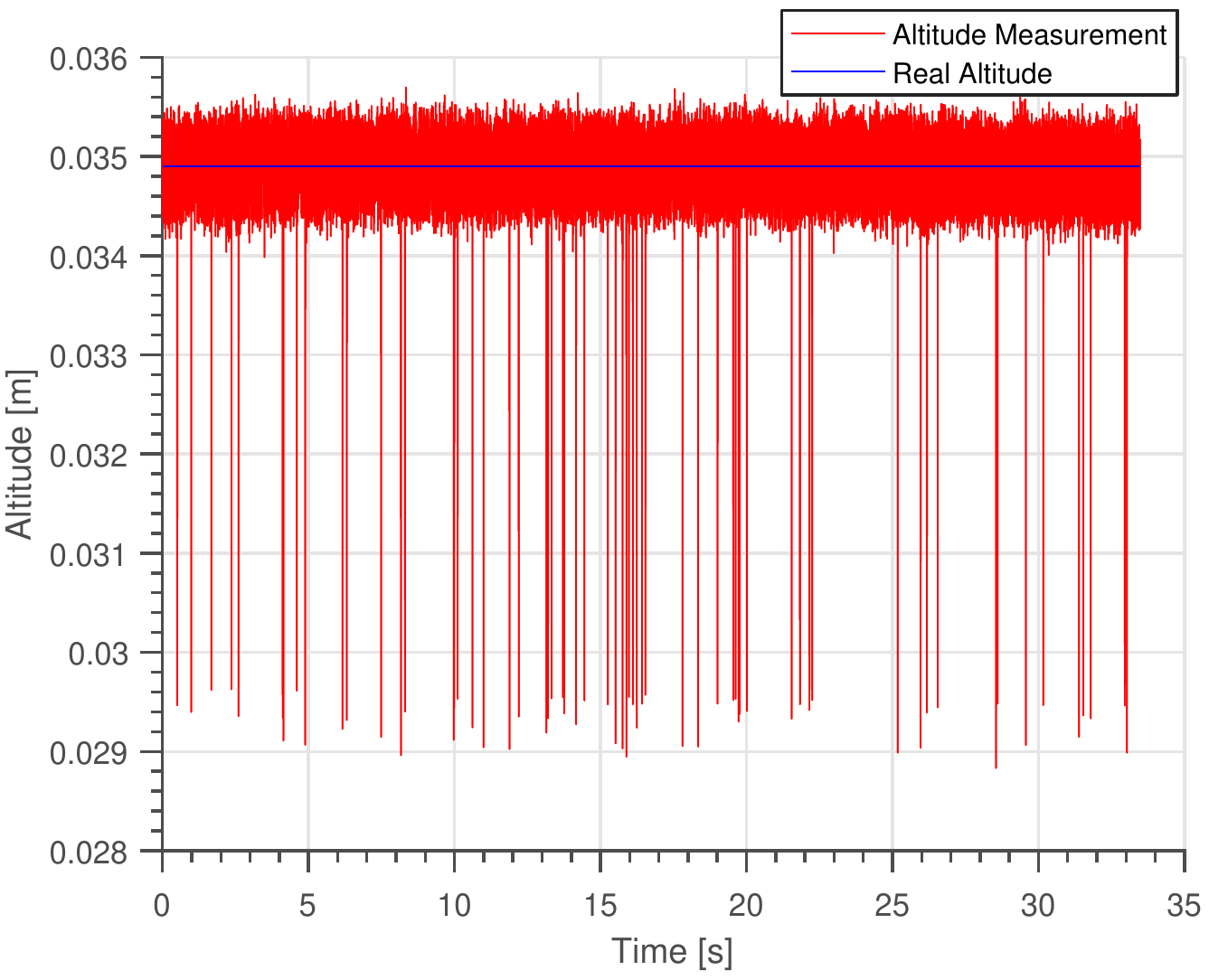}
\caption{Measurement noise while the quad-copter is stationary on the ground.}
\label{MeasurementNoise}
\end{figure}

\subsection{Model}
To obtain a single-input single-output system we kept the PID controller regulating attitude and restricted the quad-copter's motion to a vertical line. Therefore, assuming perfect attitude regulation, the quad-copter's motion can be described by:
\setlength{\arraycolsep}{0.0em}
\begin{eqnarray}
\dot{x}_1&{}={}&x_2\notag\\
\dot{x}_2&{}={}& - \frac{g}{m} + \frac{1}{m}u_{tr},\notag\\
y&{}={}&x_1,\notag
\end{eqnarray}
\setlength{\arraycolsep}{5pt}
where $x_1$ denotes altitude, $g$ is gravity's constant, and the input $u_{tr}$ represents the thrust created by the propellers rotation. The thrust is commanded by a PWM signal\footnote{We only use PWM values up to 90\% so as to leave some control authority for the attitude controller.} and the relation between the commanded PWM signal $u$ and the exerted thrust $u_{tr}$ is well described by the affine map $u_{tr}(u)=\sigma_0 + \sigma_1 u$. The input $u$ in this expression represents the fraction of the maximum allowed thrust, e.g., $u=0.6$ represents $60\%$ of the maximum thrust. This results in the dynamics:
$$\ddot{x}_1=\frac{\sigma_0 - g}{m} + \frac{\sigma_1}{m}u,$$
from which we can infer the relative degree of $y$ to be $2$ with $\alpha(x)=\frac{\sigma_0 - g}{m}$ and $\beta(x)=\frac{\sigma_1}{m}$. Since our results apply to the case where $\alpha$ and $\beta$ are functions, rather than constants, we emulate in software the functions:

\begin{equation}
\label{ExpModel}
\beta(x)=\frac{\sigma_1}{m} -  \frac{x_1^4}{2}, \qquad
\alpha(x)=\frac{\sigma_0 - g}{m} + 2\sin(x_1^2),
\end{equation}
i.e., when the data-driven controller requests the input $u$, we create the input signal $u - \frac{m}{\sigma_1} (x_1^4 u+ 4\sin(x_1^2))$. This effectively turns the control gain into a nonlinear state-dependent function. Given that $\frac{\sigma_1}{m}\approx 18$, our assumption that $\beta$ is greater than zero will be satisfied as long as the drone does not reach altitudes higher than $2.4$ meters.

In conclusion, the drone dynamics take the form:
\setlength{\arraycolsep}{0.0em}
\begin{eqnarray}
\label{DroneNonlinear}
\dot{x}_1&{}={}&x_2\notag\\
\dot{x}_2&{}={}&\frac{\sigma_0 - g}{m}+ 2\sin(x_1^2) + \left(\frac{\sigma_1}{m}-\frac{x_1^4}{2}\right)u,\notag\\
y&{}={}&x_1.
\end{eqnarray}
\setlength{\arraycolsep}{5pt}
\subsection{Data-driven controller and its implementation}
\label{SubSec:ObsAndCtrl}
%Given that we desire to control flight altitude and our actuators are the quad-copter's propellers, it is reasonable to assume the relative degree of our system is $2$. Note that our control input is actually a PWM signal defined as an integer in the interval $[0,65365]$ so encoded in $\alpha$ and $\beta$ will also be the map going from PWM signal to thrust exerted by the motors.
%the controller~\eqref{DynamicControllerFinal} using the parameter estimates provided by the observer~\eqref{ObserverFinal1}-\eqref{ProjectedObserver} and the state estimate~\eqref{StateEstimateFinal} r

The quad-copter receives altitude measurements from the motion capture system and uses them for state estimation using~\eqref{StateEstimateFinal} with $\rho=4$. This choice of $\rho$ mitigates the effects of the measurement noise that can be appreciated in Figure~\ref{MeasurementNoise}. The resulting state estimate is then fed to the controller~\eqref{DynamicControllerFinal} where $K=\left[-9\,\,-6\right]$ so as to place both eigenvalues of $A_1 + B_1K$ at $-3$ and $\gamma=0.002$. For the initial transient we use the sequence of inputs $1.0, 1.0, 1.0, 1.0$. The experiments were executed with a sample time of $T=0.0028s$ which corresponds to the maximal rate at which the motion capture system provides data.

\subsection{Experiments}
Figure \ref{Trajectory} shows how the data-driven controller successfully regulates altitude for the non-linear system \eqref{DroneNonlinear}: in the top horizontal panel we can observe the desired set-points displayed in red and the quad-copter's trajectory in blue; the second panel from the top shows that our framework achieves altitude steady-state errors consistently below $5$ millimeters; the bottom two panels portray the values taken by the state-dependent non-linear functions $\alpha(z)$ and $\beta(z)$. A comparison between the input requested by the static controller \eqref{Controller1}-\eqref{Controller2}, assuming knowledge of $\alpha$ and $\beta$, and the input generated by the dynamic controller \eqref{DynamicControllerFinal} is presented in Figure \ref{ExperimentInput}. As expected, we can see the latter converging to the former, made evident in the magnified detail.

The experimental results show that, in spite of measurement errors, the proposed data-driven controller can successfully regulate altitude. In terms of selecting the gains $K$ and $\gamma$, we can intuitively understand reductions in $\gamma$ as leading to both an increase in noise attenuation, through ``averaging'' of the estimation errors in $\widehat{z}_3$ arising from measurement noise, and a reduction to control responsiveness. Given that large gains in $K$ coupled with a small enough parameter $\gamma$ will give rise to oscillations in the system's trajectory, as often seen in systems with input delays, the authors recommend the following heuristic: start the tuning process with low control gains $K$ and a parameter $\gamma$ of the order of $T^{-1}$, judiciously increasing the gain $K$ thereafter until either a satisfactory performance is observed or oscillations arise, the latter meaning an increase in $\gamma$ might be required before continuing to increase $K$.
\begin{figure}[h]
\centering
\includegraphics[width=0.75\textwidth]{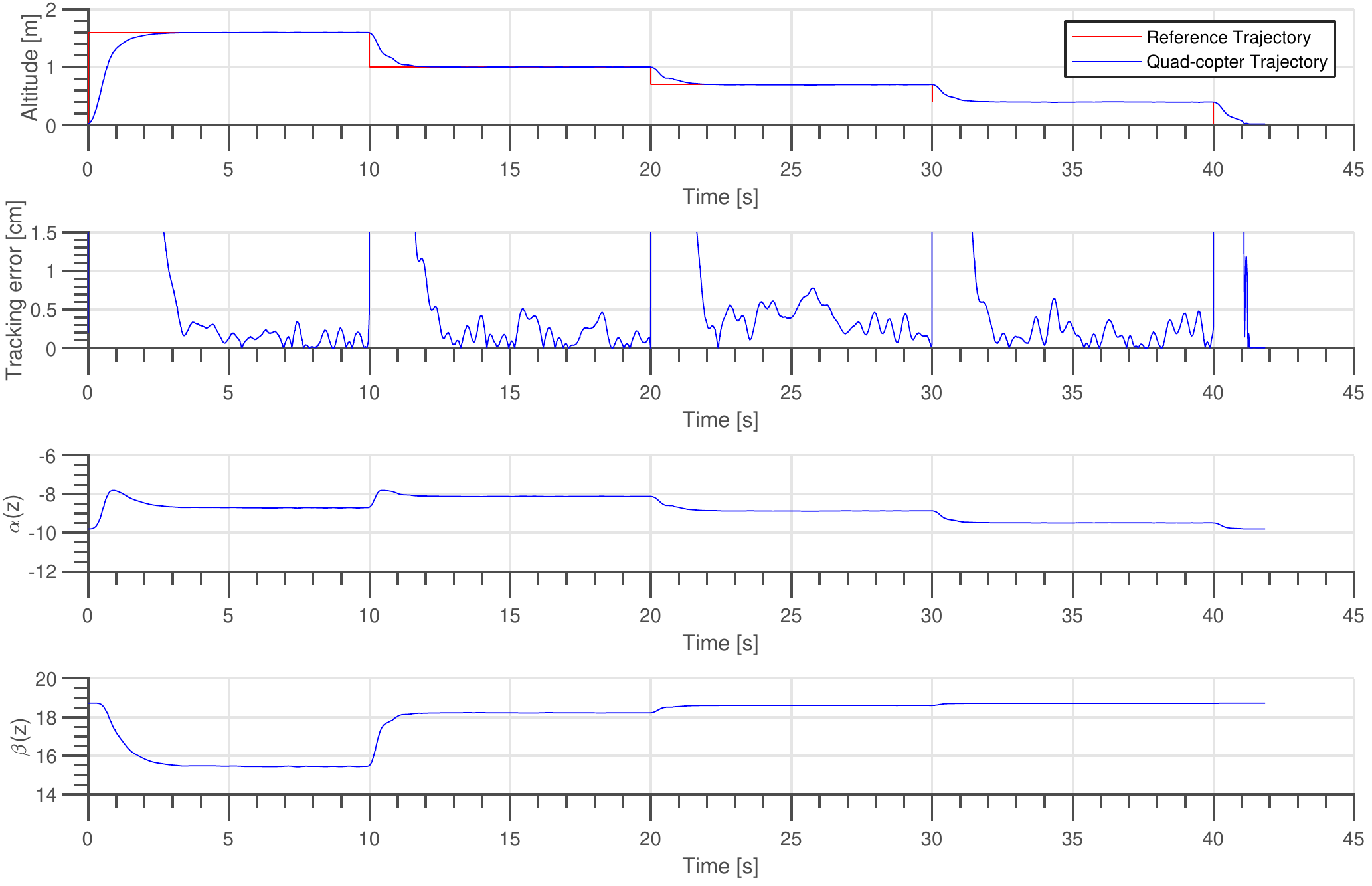}
\caption{Experimental results portraying the quad-copter's trajectory and reference trajectory, tracking error, and values of the state-dependent non-linear functions $\alpha(z)$ and $\beta(z)$.}
\label{Trajectory}
\end{figure}
\newpage
\begin{figure}[h]
\centering
\includegraphics[width=0.75\textwidth]{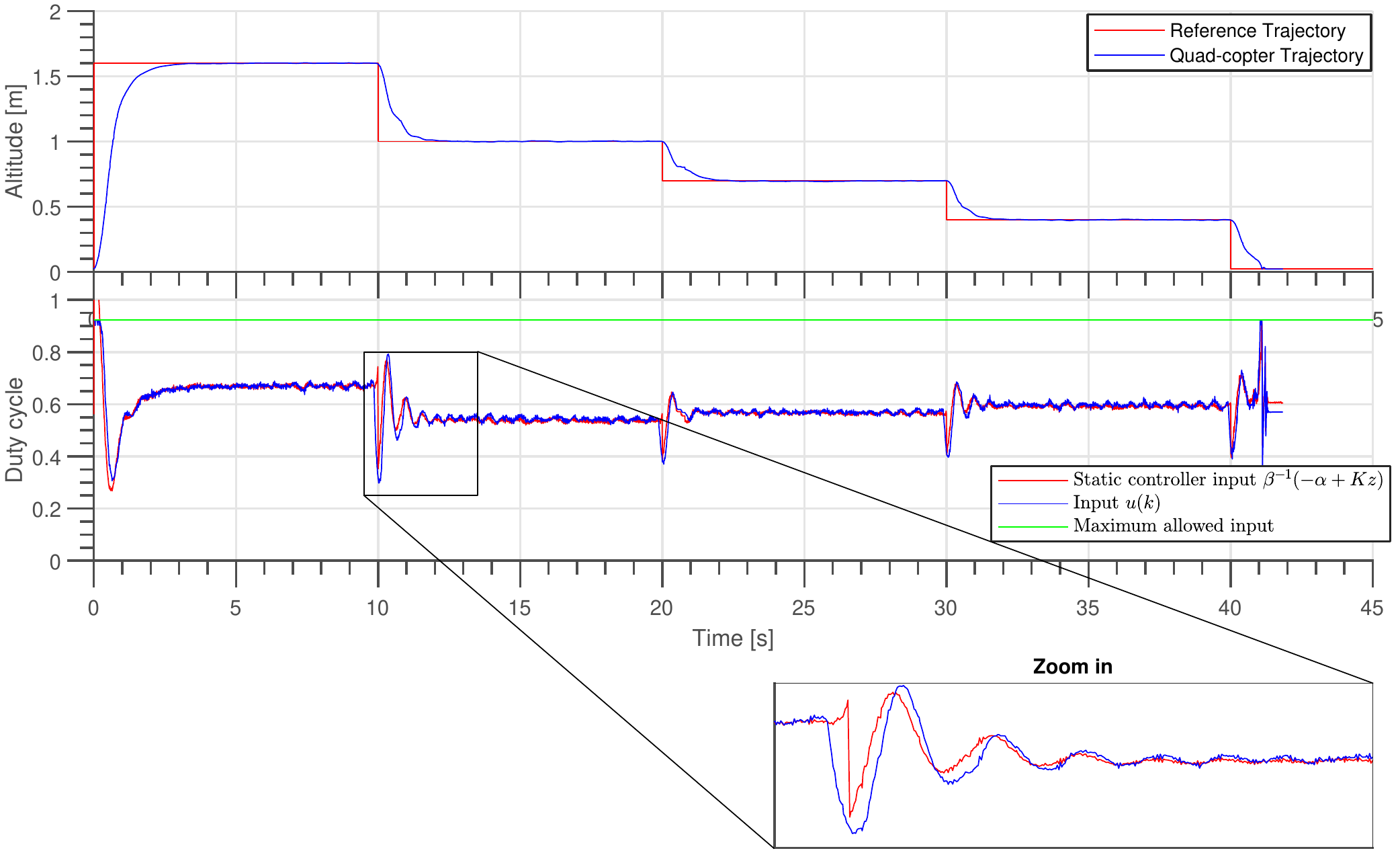}
\caption{Comparison between the input requested by the static controller \eqref{Controller1}-\eqref{Controller2}, assuming knowledge of $\alpha$ and $\beta$, and the input generated by the dynamic controller \eqref{DynamicControllerFinal}, in terms of the PWM's duty cycle. Note that the input is computed on-board the drone and reported, along with the states and parameters, to an external server at a rate of $100$ Hz to avoid draining the microprocessors resources.}
\label{ExperimentInput}
\end{figure}

\section{Conclusions}
There are several important questions that were left unaddressed in this paper. Feedback linearizability was convenient to construct the technical arguments but one can easily see extensions to partially feedback linearizable systems with well behaved zero dynamics. Similarly, extensions to the multi-input multi-output case offer no conceptual difficulties and were already discussed in Section~\ref{Sec:Main}. Identifying the largest class of systems to which the results in this paper (or suitable generalizations thereof) apply is a worthwhile endeavor. 

Equally worthwhile is investigating which state estimation and controller design techniques result in better performance in the context of the proposed data-driven methodology since it would make the results more useful in practical applications. In particular, investigating how to best mitigate the effect of measurement noise would be especially important.

Finally, there were connections made with existing results in high-gain observers, adaptive control, and potentially other areas, as well as with the recent papers~\cite{DeePC,PT20}. All of these deserve to be better understood.

\section*{Appendix}
\label{Sec:Appendix}
\begin{proof}[Proof Theorem~\ref{Theorem1}]
The proof will be based on the feedback linearized form~\eqref{UnknownSystem2a}-\eqref{UnknownSystem2c} of the dynamics rather than the original nonlinear form~\eqref{UnknownSystem1a}-\eqref{UnknownSystem1b}. This results in no loss of generality since both systems are related by the diffeomorphism $\Psi$ that satisfies $\Psi(0)=0$. For simplicity, we will denote the set $\Psi(\mathcal{S})$ simply by $\mathcal{S}$. Since $\Psi$ is a homeomorphism, $\Psi(\mathcal{S})$ is still a compact set.

\textbf{The initial transient:}
the state estimate $\widehat{z}$ requires $\rho$ samples to be collected. To simplify the argument we consider the case where $\rho=3$ which leads to an initial (fixed) sequence of $\rho - 1=2$ inputs $u_0^*,u_1^*$ used at time $k=0$ and $k=1$. This corresponds to an initial transient that must be analyzed separately. 

By applying Proposition~\ref{Prop:BoundTaylor} to the compact set $\mathcal{D}=\mathcal{S}\times\{u^*_0\}$ we conclude the existence of a time $T_0$ so that trajectories are well defined for all $T\in [0,T_0]$ and for all initial conditions in $\mathcal{S}$. We regard $T_0$ as the time elapsed during the first time step under input $u^*_0$. The set of points reached under all these trajectories and for all $T\in [0,T_0]$ is denoted by $Z_0$. We can repeat this argument, using $Z_0$ as the set of initial conditions (and assuming the initial time to be zero) and the input $u_1^*$ to conclude the existence of a time $T_1$ so that trajectories are well defined for all $T\in [0,T_1]$ and for all initial conditions in $Z_0$. By taking $T_2=\min\{T_0,T_1\}$ we conclude that solutions are well defined for the sequence of inputs $u_0^*,u_1^*$ where each input is applied for $T_2$ units of time. Let now $Z$ be the set of points reached under all the trajectories with initial conditions in $\mathcal{S}$ and that result by applying $u_0^*$ for $T$ units of time, followed by applying $u_1^*$ for $T$ units of time with $T$ ranging through all the values in the set $[0,T_2]$. This set will be used several times in the remainder of the proof.

At time step $k=2$ the state estimate is readily available. Let us denote by $E$ the set of possible values taken by the dynamic controller's state $e_u$ at time $k=2$ depending on the different initial conditions and the chosen constants $u_0^*$, $u_1^*$. Since the input error $e_u$ is a continuous function of the initial condition $z(0)$ that belongs to the compact set $\mathcal{S}$ and the constants $u_0^*$, $u_1^*$, $E$ is a bounded set. Consider also the set $R$ defined as the smallest sub-level set of $W=V_z + V_{e_u}$ that contains $Z\times E$ where $V_{e_u}:\R \to \R$ is the Lyapunov function defined by $V_{e_u}(s) = s^2$ and $V_z$ is the Lyapunov function satisfying~\eqref{DefineV}. Our objective is to show that $R$ is an invariant set.

 \textbf{Existence of solutions one step beyond the transient:} we first show that it is possible to continue the solutions from $R$ by employing again Proposition~\ref{Prop:BoundTaylor}. For future use, we define the projections $\pi_1:\R^2\to\R$, $\pi_Z:\R^2\times\R\to \R^2$, and $\pi_E:\R^2\times\R\to \R$  defined by $\pi_1(z_1,z_2)=z_1$, $\pi_Z(z,e_u)=z$, and $\pi_E(z,e_u)=e_u$. The dynamic controller is a function of $\widehat{z}$, however, since $\widehat{z}$ is a function of $z$, we can regard the controller as a smooth function of $z$. We can thus consider the set of inputs $U\subset \R$ defined by all the inputs obtained via our dynamic controller when $u=u_1^*$, $z$ ranges in $\pi_Z(R)$ and $\widehat{z}$ is given by~\eqref{StateEstimateFinal} with $Y(2)$, defined in~\eqref{DefY}, ranging in $\left(\pi_1\circ \pi_Z(R)\right)^{\rho}$, the $\rho$-fold Cartesian product of $\pi_1\circ\pi_Z(R)$. By taking its closure, if needed, we can assume the set $\pi_Z(R)\times U$ to be compact and apply Proposition~\ref{Prop:BoundTaylor} to obtain a time $T_3$ ensuring that solutions starting at $\pi_Z(R)$ exist for all $T\in [0,T_3]$. Moreover, Proposition~\ref{Prop:BoundTaylor} ensures the existence of a constant $M$ for which the bound~\eqref{BoundInProp} holds and, as a consequence, the approximate model~\eqref{ModelOrder33a}-\eqref{ModelOrder33b} is valid for any solution with initial condition in $R$ and any input in $U$. If $T_2<T_3$ we proceed by only considering sampling times in $[0,T_2]$ and note that none of the conclusions reached so far change. If $T_2>T_3$, we can use sampling times in $[0,T_3]$ while noting that all the reached conclusions remain valid by redefining $Z$ to be the set of points reached for any time in $[0,T_3]$ (if the conclusions hold for the (non-strictly) larger $Z$ set they also hold for the (non-strictly) smaller $Z$ obtained by reducing $T_2$ to $T_3$).

\textbf{Invariance of the set $R$:} we can now establish invariance of $R$ by computing \mbox{$W(F_T^e(z,u),G_T^e(e_u))-W(z,e_u)$} for all $(z,e_u)\in R$ with $u$ given by~\eqref{DynamicControllerFinal} and $G_T^e$ denoting the exact dynamics of the input error $e_u$. We will do this in several steps.

In the first step we establish that the evolution of $V_z$ under $F_T^e$ equals the evolution of $V_z$ under $F_T^a$ up to $O(T^2)$ terms. In order to do so, we recall that $F_T^e(z,u)$ can be expressed as: 
\setlength{\arraycolsep}{0.0em}
\begin{eqnarray}
F_T^e(z,u)&{}={}&F_T^a(z,u) + O_{(z,u - u _0)}(T^2)\notag\\
&{}={}&Az + B(\alpha + \beta u) + O_{(z,u - u _0)}(T^2)\notag\\
\end{eqnarray}
\setlength{\arraycolsep}{5pt}
We then have:
\setlength{\arraycolsep}{0.0em}
\begin{eqnarray}
V_z(F_T^e(z,u)) - V_z(z) &{}={}& (F_T^a(z,u) + O_{(z,u - u _0)}(T^2))^TP_z(F_T^a(z,u) + O_{(z,u - u _0)}(T^2))-z^TP_z z\notag\\
&{}={}& V_z(F_T^a(z,u)) - V_z(z) + 2O_{(z,u - u _0)}^T(T^2)P_z  F_T^a(z,u)\notag \\
&&{+}\:O_{(z,u - u _0)}^T(T^2)P_z O_{(z,u - u _0)}(T^2)\notag\\
&\le& V_z(F_T^a(z,u)) - V_z(z) + 2O_{(z,u - u _0)}^T(T^2)P_z  F_T^a(z,u)+O_{(z,u - u _0)^2}(T^4)\notag\\
&{}={}& V_z(F_T^a(z,u)) - V_z(z) + 2O_{(z,u - u _0)}^T(T^2)P_z  Az\notag \\
&&{+}\:2O_{(z,u - u _0)}^T(T^2)P_z B (\alpha + \beta u) + O_{(z,u - u _0)^2}(T^4)\notag\\
&{}={}& V_z(F_T^a(z,u)) - V_z(z) + O_{(z,u - u _0)^2}(T^2) + O_{(z,u - u _0)^2}(T^3)+O_{(z,u - u _0)^2}(T^4)\notag\\
\label{LyapunovIneqV}
&{}={}& V_z(F_T^a(z,u)) - V_z(z) + O_{(z,u - u _0)^2}(T^2),
\end{eqnarray}
\setlength{\arraycolsep}{5pt}
where we used the relationship \mbox{$\Vert z\Vert\le \Vert (z,u - u _0)\Vert$} and boundedness of $\alpha$, $\beta$, and $u$ in virtue of $(z,u)$ belonging to the compact set $\pi_Z(R)\times U$, to obtain the fourth equality.

Noting that from the definition of $e_u$, equation \eqref{eu}, one can reach the expression:
\begin{equation}
\label{ControllerFinalExpression}
u=\beta^{-1}( - \alpha + v(z))-\beta^{-1}e_u,
\end{equation}
we consider the term $\Vert (z,u - u _0)\Vert$ in more detail,
	\setlength{\arraycolsep}{0.0em}
	\label{HandlingInput1}
	\begin{eqnarray}
	\Vert (z,u - u_0)\Vert &{}\leq{}& \Vert z \Vert + \left\Vert \beta^{-1}( - \alpha + v(z))-\beta^{-1}e_u - u_0 \right\Vert\notag\\
	&{}\leq{}&\Vert z \Vert + \left\Vert \beta^{-1}(- \alpha + v(z)) - u_0\right\Vert + \left\Vert \beta^{-1}e_u\right \Vert.
	\end{eqnarray}
	\setlength{\arraycolsep}{5pt}
As the function $\beta^{-1}(z)(- \alpha(z) + v(z))$ is Lipschitz continuous (with Lipschitz constant $L$) on $\pi_Z(R)$, and it produces the value $u_0$ at $z=0$, we conclude that: 
\begin{equation}
\label{HandlingInput2}
\left\Vert\beta^{-1}\left( - \alpha + v(z)\right) - u_0\right\Vert\le L\Vert z\Vert.
\end{equation}
The preceding sequence of inequalities, and boundedness of $\beta^{-1}$ on $\pi_Z(R)$, lead to the useful expression:
\begin{equation}
\label{BoundEstError2}
O_{(z,u - u _0)}(T)=O_z(T) + O_{e_u}(T).
\end{equation}
Combining the previous bounds~\eqref{LyapunovIneqV} and~\eqref{BoundEstError2} we obtain:
\setlength{\arraycolsep}{0.0em}
\begin{align*}
\label{Step1}
V_z(F_T^e(z,u)) - V_z(z){}={} V_z(F_T^a(z,u)) - V_z(z) + O_{z^2}(T^2) + O_{e_u^2}(T^2), \numberthis
%&\le V(F_T^a(z,u)) - V(z) + cdM^2{T}^6\left(\Vert z\Vert^2 + \Vert e_z\Vert^2 + \Vert\Phi^Te_{\pi}\Vert^2\right),\\
%\label{FinalIneq1}
%&{}={}& V(F_T^a(z,u)) - V(z) + O_{z^2}(T^6) + O_{e_z^2}(T^6) + O_{(\Phi^T e_\pi)^2}(T^6)\\
\end{align*}
\setlength{\arraycolsep}{5pt}
%\begin{eqnarray}
%\label{Step1}
%V(F_T^e(z,u)) - V(z)&\le& V(F_T^a(z,u)) - V(z) + O_{z^2}(T^2) + O_{c}(T^2).
%\end{eqnarray}
%Note that, in particular, we established that $O_{(z,u - u _0)}(T)$ can be replaced with $O_z(T) + O_{e_z}(T) + O_{\Phi^T e_\pi}(T)$. We will be using this fact in the remainder of the proof without explicitly acknowledging it.
establishing that the decrease of $V_z$ imposed by $F^e_T$ equals the decrease imposed by $F^a_T$ up to $O(T^2)$ terms.

%where $d\in \R^+ $ is a suitable constant and we used the fact that $\beta^{-1}$ is bounded on the compact $\pi_Z(R)$ and equality~\eqref{StateEstimationError} combined with the previous argument to replace $O_{e_z^2}(T^6)$ with $O_{z^2}(T^6)$. Note that, in particular, we established that $O_{(z,u - u _0)}(T)$ can be replaced with $O_z(T) + O_{\Phi^T e_\pi}(T)$. We will be using this fact in the remainder of the proof without explicitly acknowledging it.

In the second step we use the definition of $e_u$ in the form given by \eqref{ControllerFinalExpression} to show that $V_z(F_T^a(z,u)) - V_z(z)$ is negative definite up to $O_{z^2}(T^2)$ and $O_{e_u^2}(T)$ terms. Using expression \eqref{ControllerFinalExpression}, the approximate dynamics are given by:
\setlength{\arraycolsep}{0.0em}
\begin{eqnarray}
F^a_T(z,u)&{}={}& Az+Bv(z)-Be_u. \notag
\end{eqnarray}
\setlength{\arraycolsep}{5pt} 
We can now compute $V_z(F_T^a(z,u)) - V_z(z)$ as:
\setlength{\arraycolsep}{0.0em}
\begin{align*}
\label{Step2}
	V_z(F_T^a(z,u))& - V_z(z) \\
	&={} (Az+Bv(z))^TP_z(Az+Bv(z)) - V_z(z) \notag \\
	&\qquad-e_uB^TP_z(Az+Bv(z))-(Az+Bv(z))^TP_zBe_u+B^TP_zBe_u^2\\
	&={} V_z(F_T^a(z,\overline{u})) - V_z(z) \notag \\
	&\qquad-e_uB^TP_z(Az+Bv(z))-(Az+Bv(z))^TP_zBe_u+B^TP_zBe_u^2\\
	&{}\le{} - \lambda_{\min}(Q)T\Vert z\Vert^2 + O_{(z,u - u _0)^2}(T^2)+2T \Vert (A + BK) P_z(B_1+B_2T)\Vert\Vert z\Vert\Vert e_u\Vert\notag \\
	&\qquad+T^2(B_1+TB_2)^TP_z(B_1+TB_2)e_u^2\notag \\
	&{}\le{} - \frac{\lambda_{\min}(Q)}{2}T\Vert z\Vert^2 + c T e_u^2+O_{e_u^2}(T^2)+ O_{(z,u - u _0)^2}(T^2)\notag \\
	&{}\le{} - \lambda_z T\Vert z\Vert^2+O_{z^2}(T^2)+O_{e_u^2}(T),\numberthis
\end{align*}
\setlength{\arraycolsep}{5pt} 
where we reach: the second equality due to \eqref{AffineModel}; the first inequality due to \eqref{DefineV}; the second inequality, which holds for any $c\in\R$ satisfying $c>\frac{2}{\lambda_{\min}(Q)}\norm{(A + BK)P(B_1+B_2T)}^2$, by completing squares; and the last inequality by using equality \eqref{BoundEstError2} and selecting $\lambda_z\in\R^+$ satisfying $\lambda_z\le\frac{\lambda_{\min}(Q)}{2}$.

In the third step we analyze the effect of using the estimates $\widehat{z}$ and $\widehat{z}_3$ when implementing the control law \eqref{DynamicControllerFinal} by substituting $\widehat{z}=z+e_{z_{(1,2)}}$ and $\widehat{z}_3=z_3+e_{z_{(3)}}$, where $e_{z_{(1,2)}}$ represents the vector composed of the first two entries of $e_z$ and $e_{z_{(3)}}$ represents its third entry, and evaluating the dynamics of the error $e_u$. Based on the relation \eqref{StateEstimationError}, the control law \eqref{DynamicControllerFinal} can be expressed as
\begin{align*}
	\label{ControllerFinalError}
	u(k+1)&=u+\gamma \left(v(\widehat{z})-\widehat{z}_3\right)\\
	&=u+\gamma \left(v(z)+Ke_{z_{(1,2)}}-\left(\alpha+\beta u\right)+e_{z_{(3)}}\right)\\
	&=u+\gamma \left(v(z)-\left(\alpha+\beta u\right)\right)+Ke_{z_{(1,2)}}+e_{z_{(3)}}\\
	&=u+\gamma e_u+O_{(z,u-u_0)}(T),\numberthis
\end{align*}
Before going forward, we apply Proposition~\ref{Prop:BoundTaylor} to $(\alpha,\beta)\circ F^e_{T}(z,u)$ to obtain
\begin{equation}
\label{ApproxParameters}
\alpha(T)=\alpha(0) + O_{(z,u - u_0)}(T),\quad \beta(T)=\beta(0) + O_{(z,u - u_0)}(T).
\end{equation}
With these equalities at hand, we compute $G^e_T(e_u)=e_u(k+1)$:
\begin{align*}
\label{Nexteu1}
G^e_T(e_u(k))&=v(z(k+1))-z_3(k+1)\\
	&=v\left(z+O_{(z,u-u_0)}(T)\right)-\left(\alpha (k+1)+\beta (k+1) u(k+1)\right)\\
	&=v(z(k))+KO_{(z,u-u_0)}(T)-\left(\alpha(k)+\left(\beta(k)+O_{(z,u-u_0)}(T)\right)u(k+1)+O_{(z,u-u_0)}(T)\right)\\
	&=v(z(k))+O_{(z,u-u_0)}(T)-\left(\alpha(k)+\beta(k) u(k+1)\right)+u(k+1)O_{(z,u-u_0)}(T)\\
	&=v(z(k))-\left(\alpha(k)+\beta(k) u(k)\right)- \beta(k)\gamma e_u(k)-\beta(k)O_{(z,u-u_0)}(T)\\
	&\quad+\left(u(k)+\gamma e_u(k)+O_{(z,u-u_0)}(T)\right)O_{(z,u-u_0)}(T)+O_{(z,u-u_0)}(T)\\
	&=\left(1 - \beta(k)\gamma \right)e_u(k)+O_{(z,u-u_0)}(T)+O_{(z,u-u_0)^2}(T^2)\\
	&=\left(1 - \beta(k)\gamma \right)e_u(k)+O_{(z,u-u_0)}(T), \numberthis
\end{align*}
where we reach: the second equality by using the exact model \eqref{ModelOrder3a}-\eqref{ModelOrder3b} with $t=T$ aggregating all terms with a $T$ coefficient inside $O_{(z,u-u_0)}(T)$, and the definition of $z_3$ from Section~\ref{Sec:StateEstimation}; the third equality by using \eqref{ApproxParameters}; the fifth equality by substituting $u(k+1)$ with~\eqref{DynamicControllerFinal} and using the definition of $e_u(k)$ from~\eqref{eu}; the sixth equality by absorbing $u$ and $e_u$ into the $O_{(z,u-u_0)}(T)$ term on account of $e_u$ and $u$ being bounded in $R$; and the last equality by noting that $O_{(z,u-u_0)^2}(T^2)=O_{(z,u-u_0)}(T)$ on account of $z$ and $u$ belonging in the compact sets $R$ and $U$.

 Coming back to the Lyapunov function $V_{e_u}(e_u)=e_u^2$, we can now compute $V_{e_u}(G^e_T(e_u))-V_{e_u}(e_u)$:
\begin{align*}
\label{DefineVu}
	V_{e_u}(G^e_T(e_u))-V_{e_u}(e_u)&=\left(\left(1-\beta\gamma\right)e_u+O_{(z,u-u_0)}(T)\right)^2-e_u^2\\
	&=\left(-2\beta\gamma+\beta^2\gamma^2 \right)e_u^2+2\left(1-\beta\gamma\right)e_u O_{(z,u-u_0)}(T)+O_{(z,u-u_0)^2}(T^2)\\
	&\le \left(-\beta\gamma+\beta^2\gamma^2 \right) e_u^2+cO_{(z,u-u_0)^2}(T^2)\\
	&\le -\lambda_u e_u^2+O_{(z,u-u_0)^2}(T^2)\\
	&\le -\lambda_u e_u^2+O_{z^2}(T^2)+O_{e_u^2}(T^2),\numberthis
\end{align*}
where: we reach the first inequality, which holds for any $c\in\R$ satisfying $c>1+\left(1-\beta\gamma\right)^2/\left(\beta\gamma\right)$, by completing squares; the second inequality holds for sufficiently small\footnote{In particular, this inequality holds for any $\gamma$ and $\lambda_u$ satisfying $\gamma\leq \overline{\beta}^{-1}$ and $\lambda_u\leq \left(1-\underline{\beta}\gamma\right)\underline{\beta}\gamma$, where $\overline{\beta}=\max_{z\in R}\beta(z)$ and $\underline{\beta}=\min_{z\in R}\beta(z).$}  $\gamma,\lambda_u \in \R^+$; we reach the last inequality by using equality \eqref{BoundEstError2}.
We now put the three intermediate steps,~\eqref{Step1} and ~\eqref{Step2}, and \eqref{DefineVu} together:
\setlength{\arraycolsep}{0.0em}
\begin{align*}
	W(F_T^e(z,u),G_T^e(e_u)) - W(z,e_u)&\le  - \lambda_z T\Vert z\Vert^2 - \lambda_u\Vert e_u \Vert^2 + O_{z^2}(T^2)+ O_{e_u^2}(T)\notag\\
	&\le  - \lambda_z T\Vert z\Vert^2 - \lambda_u \Vert e_u\Vert^2 + MT^2\Vert z\Vert^2 + M T\Vert e_u \Vert^2.\notag
\end{align*}	
where $M\in\R^+$ is the largest constant stemming from the definition of the $O$ terms.  If we choose $\lambda\in \R^+$ and $T_4\in \R^+$ satisfying:
\setlength{\arraycolsep}{0.0em}
\begin{eqnarray}
\lambda &{}<{}&\min\{\lambda_z,\lambda_{u}\},\notag \\
T_4&{}<{}&\min\left\{\frac{(\lambda_{z}- \lambda)}{M},\frac{(\lambda_{u}- \lambda)}{M}\right\},\notag
\end{eqnarray}
\setlength{\arraycolsep}{5pt}
it follows that for all $T\in [0,T_4]$ we have:
\setlength{\arraycolsep}{0.0em}
\begin{eqnarray}
\label{LastInequality}
W(F_T^e(z,u),G_T^e(e_u)) - W(z,e_u) &\le&  - \lambda T\Vert z\Vert^2 - \lambda \Vert e_u\Vert^2.
\end{eqnarray}
\setlength{\arraycolsep}{5pt}
Therefore, for any $T\in [0,T_5]$, $T_5=\min\{T_1,\hdots,T_4\}$, we have that $R$ remains invariant. By noting that trajectories remain in $R$ for any time in $[0,T_5]$ we conclude that we can apply the same argument to establish that trajectories remain in $R$ for any number of time steps since we only assumed that inputs were generated based on output measurements that remained in $\pi_1\circ\pi_Z(R)$. Compactness of $R$ establishes that trajectories are bounded and thus there exists a constant $b_1\in \R^+$ so that $\Vert e_u(k)\Vert\le b_1$ and $\Vert z(k)\Vert \le b_1$ for all $k\in \N$. Moreover,~\eqref{LastInequality} informs us that both $z$ and $e_u$ will converge to the origin. Invoking Theorem~1 in~\cite{NTS99}, combined with invariance of $R$ and smoothness of the dynamics, we conclude that the solutions of~\eqref{UnknownSystem1a}, when using the dynamic controller~\eqref{DynamicControllerFinal} , where the virtual input $v$ is provided by \eqref{Controller2}, using the state estimates provided by an estimation technique satisfying~\eqref{StateEstimationError}, are bounded, i.e., there exists a constant $b_2\in \R^+$ so that $\Vert x(t)\Vert\le b_2$ and, moreover,  $\lim_{t\to \infty}x(t)=0$. Hence, by taking $b=\max\{b_1,b_2\}$ we conclude the proof.
\end{proof}

\begin{proof}[Proof of Theorem~\ref{Lemma1}]
It is sufficient to verify that the sections of the proof of Theorem~\ref{Theorem1} that depend on $v(z)$ hold under any virtual input satisfying the conditions \ref{lemma1req1}-\ref{lemma1req2}, the rest remains unchanged. An attentive reader will notice that only equations \eqref{Step2}, \eqref{ControllerFinalError} and \eqref{Nexteu1} could be affected by a change in $v(z)$. That being said, it is straightforward to see that if $v(z)$ is such that if condition \eqref{lemma1req1} is satisfied, then \eqref{Step2} remains unchanged, and similarly, that if \eqref{lemma1req2} holds then equation \eqref{ControllerFinalError} and inequality \eqref{Nexteu1} hold. Thus, we conclude that Theorem \ref{Theorem1} holds when replacing $v(z)$ as provided in \eqref{Controller2} with any $v(z)$ satisfying conditions \eqref{lemma1req1} and \eqref{lemma1req2}.
\end{proof}

\begin{proof}[Proof of Theorem~\ref{Theorem2}]
As stated in Section~\ref{Sec:StateEstimation}, in the presence of essentially bounded noise the estimation error is given by $e_z=O_{(z,u - u _0)}(T) + O_{\, \overline{d}\,}(T^{-n})$. This proof follows the same arguments of the proof of Theorem~\ref{Theorem1} while accounting for the effect of measurement noise in $e_z$. Therefore, we shall describe only the required modifications.

Due to measurement noise, we redefine $E$ as the set of possible values taken by the dynamic controller's state $e_u$ at time $k=2$ depending on the different initial conditions,  and the chosen constants $u_0^*$, $u_1^*$. Since the input error $e_u$ is a continuous function of the initial condition $z(0)$ that belongs to the compact set $\mathcal{S}$ and the constants $u_0^*$, $u_1^*$, $E$ is a bounded set. Consider also the set $R$ defined as the smallest sub-level set of $W=V_z + V_{e_u}$ that contains $Z\times E$ where $V_{e_u}:\R \to \R$ is the Lyapunov function defined by $V_{e_u}(s) = s^2$ and $V_z$ is the Lyapunov function satisfying~\eqref{DefineV}. Our objective is to show that $R$ is an invariant set.

 \textbf{Existence of solutions one step beyond the transient:} we first show that it is possible to continue the solutions from $R$ by employing again Proposition~\ref{Prop:BoundTaylor}. For future use, we define the projections $\pi_1:\R^2\to\R$, $\pi_Z:\R^2\times\R\to \R^2$, and $\pi_E:\R^2\times\R\to \R$  defined by $\pi_1(z_1,z_2)=z_1$, $\pi_Z(z,e_u)=z$, and $\pi_E(z,e_u)=e_u$. The dynamic controller is a function of $\widehat{z}$, however, since $\widehat{z}$ is a function of $z$, we can regard the controller as a smooth function of $z$. We can thus consider the set of inputs $U\subset \R$ defined by all the inputs obtained via our dynamic controller when $u=u_1^*$, $z$ ranges in $\pi_Z(R)$ and $\widehat{z}$ is given by~\eqref{StateEstimateFinal} with $Y(2)$, defined in~\eqref{DefY}, ranging in $\left(\pi_1\circ \pi_Z(R)\right)^{\rho}$, the $\rho$-fold Cartesian product of $\pi_1\circ\pi_Z(R)$. By taking its closure, if needed, we can assume the set $\pi_Z(R)\times U$ to be compact and apply Proposition~\ref{Prop:BoundTaylor} to obtain a time $T_3$ ensuring that solutions starting at $\pi_Z(R)$ exist for all $T\in [0,T_3]$. Moreover, Proposition~\ref{Prop:BoundTaylor} ensures the existence of a constant $M$ for which the bound~\eqref{BoundInProp} holds and, as a consequence, the approximate model~\eqref{ModelOrder33a}-\eqref{ModelOrder33b} is valid for any solution with initial condition in $R$ and any input in $U$. If $T_2<T_3$ we proceed by only considering sampling times in $[0,T_2]$ and note that none of the conclusions reached so far change. If $T_2>T_3$, we can use sampling times in $[0,T_3]$ while noting that all the reached conclusions remain valid by redefining $Z$ to be the set of points reached for any time in $[0,T_3]$ (if the conclusions hold for the (non-strictly) larger $Z$ set they also hold for the (non-strictly) smaller $Z$ obtained by reducing $T_2$ to $T_3$).

Due to measurement noise, the subset $U \subset \R$ is now defined by all the inputs obtained via our dynamic controller when $u=u_1^*$, $z$ ranges in $\pi_Z(R)$, $d \in [-\overline{d}, \overline{d}]$ and $\widehat{z}$ is given by~\eqref{StateEstimateFinal} with $Y(2)$, defined in~\eqref{DefY}, ranging in $\left(\pi_1\circ \pi_Z(R)\right)^{\rho}$, the $\rho$-fold Cartesian product of $\pi_1\circ\pi_Z(R)$. Given that $U$ is still compact we can use the same arguments as in Theorem~\ref{Theorem1} to guarantee existence of solutions one step beyond the transient.

We now note that to establish boundedness of all the signals it is sufficient to establish the existence of a sub-level set $R$ of $W$ that is forward invariant and satisfies $\mathcal{S}\subseteq\pi_Z(R)$. As in the previous proof we define $R$ to be the smallest sub-level set of $W=V_z + V_{e_u}$ that contains $Z\times E$ at the end of the initial transient. Given that $R$ is a compact set we can define $c_1=\max_{(z,e_u)\in R}\Vert(z,u-u_0)\Vert$ and $c_2=\min_{(z,e_u) \in \delta R}\sqrt{T\Vert z \Vert^2+\Vert e_u \Vert^2}$, where $\delta R$ is the boundary of the set $R$. Note that $c_2$ is greater than zero as the origin is assumed to be contained in the interior of $\mathcal{S}$ which is itself contained in the interior of $R$. 

Under measurement noise $d$, equality \eqref{ControllerFinalError} becomes:
\begin{equation}
u(k+1)=u(k)+\gamma e_u+O_{(z,u-u_0)}(T)+O_{\overline{d}^m}(T^{-n}).
\end{equation}
This in turn results in $G^e_T(e_u)=e_u(k+1)$ becoming:
\begin{align*}
\label{NexteuError}
	G^e_T(e_u(k))=\left(1 - \beta(k)\gamma \right)e_u(k)+O_{(z,u-u_0)}(T)+O_{\overline{d}}(T^{-n})\numberthis
\end{align*}
Based on this equality, it can be shown that:
\begin{align*}
\label{DefineVuError}
	V&_{e_u}(G^e_T(e_u))-V_{e_u}(e_u)\le -\lambda_u e_u^2+O_{z^2}(T^2)+O_{e_u^2}(T^2)+O_{\overline{d}^2}(T^{-2n})\numberthis
\end{align*}
Inequality \eqref{DefineVuError} allows us to conclude that:
\begin{align*}
	W(F_T^e(z,u),G_T^e(e_u)) - W(z,e_u)	&\le - \lambda T\Vert z\Vert^2 - \lambda\Vert e_u\Vert^2 + M\overline{d}^{2}T^{-2n},\\
	&\le - \lambda c_2^2 + M\overline{d}^{2}T^{-2n}.
\end{align*}	
Thus, $$W(F_T^e(z,u),G_T^e(e_u)) - W(z,e_u)<0$$
holds for all $(z,e_u)\in\delta R$ and $\overline{d}<b_1=\left(\sqrt{\frac{\lambda}{M}}c_2\right)^{\frac{1}{2}}T^n$, showing that $R$ is invariant. This guarantees that all signals remain bounded, i.e., if we define $b_2'$ as the radius of the smallest ball containing $R$ we conclude that $\Vert \widehat{z}(k)\Vert\le b_2'$ and $\left\Vert e_u\right\Vert\le b_2'$ for all $k\in \N$. By using arguments similar to those employed in the proof of Theorem~\ref{Theorem1}, there exists a constant $b_2''$ so that $\Vert x(t)\Vert\le b_2''$ for all $t\in \R$ and we can define $b_2$ to be $\max\{b_2',b_2''\}$.

Moreover, trajectories will converge to the smallest sub-level set of $W$ containing the ball of radius \mbox{$r$} centered at zero where $r$ is the smallest real number satisfying $ - \lambda r^2 + M\overline{d}^{2}T^{-2n}\le 0$, i.e., $r=\overline{d}\, {T}^{ -n}\left(\frac{\lambda}{M}\right)^{-\frac{1}{2}}$. Since said sub-level set is contained in the ball centered at the origin and of radius $r\lambda_{\max}(P_w)/\lambda_{\min}(P_w)$ where $P_w$ is the matrix defining the quadratic Lyapunov function $W(z,e_u)=w^T P_w w$, $w=(z,e_u)$, the result is proved by taking: $$b_3=\left(\frac{\lambda}{M}\right)^{-\frac{1}{2}}\frac{\lambda_{\max}(P_w)}{\lambda_{\min}(P_w)}.$$
\end{proof}

\color{black}

%
%The inequality $\vert V(a) - V(b)\vert\le c\Vert a - b\Vert^2$ is equivalent to (we assume $V(a)>V(b)$ without loss of generality):
%$$\sup_{a,b}\phi(a,b)=\sup_{a,b}\frac{V(a) - V(b)}{(a - b)^T(a - b)}\le c.$$
%We fist note that $\phi(a,b)$ is homogeneous, i.e., $\phi(\delta a,\delta b)=\phi(a,b)$ for any $\delta\in \R$. Hence, it suffices to compute the supremum for vectors $(a,b)$ satisfying $(a - b)^T(a - b)=1$, i.e., it suffices to compute $\sup_{a,b} V(a) - V(b)$. Since $V(a) - V(b)$ is a continuous function and $a,b$ range in a compact set, a finite upper bound exists.

\bibliographystyle{IEEEtran}
\bibliography{IEEEabrv,arXivFinal}
\end{document}